\documentclass{article}

\usepackage{PRIMEarxiv}

\usepackage[utf8]{inputenc} 
\usepackage[T1]{fontenc}    

\usepackage{amsmath, amssymb, amsfonts}
\usepackage{amsthm}
\usepackage{bm}

\newtheorem{lemma}{Lemma}
\newtheorem{theorem}{Theorem}

\newtheorem{definition}{Definition}

\newtheorem{problem}{Problem}

\renewcommand{\P}{\mathcal{P}}
\newcommand{\diag}{\operatorname{diag}}

\usepackage{booktabs, multirow, array}
\usepackage{siunitx}

\usepackage{algorithm}
\usepackage{algpseudocode} 

\usepackage{graphicx}
\graphicspath{{./Figures/}{./images/}{./media/}}
\usepackage{subcaption}
\usepackage[table,xcdraw]{xcolor}
\definecolor{corabg}{RGB}{240,248,255}
\definecolor{citebg}{RGB}{255,240,245}
\definecolor{pubmbg}{RGB}{240,255,240}

\usepackage[most]{tcolorbox}
\newtcolorbox{graybox}{
    colback=gray!6,
    colframe=black!30,
    arc=2mm,
    boxrule=0.4pt
}

\usepackage{url}
\usepackage{microtype}
\usepackage{fancyhdr}
\usepackage[switch]{lineno}

\usepackage{tikz}
\usetikzlibrary{positioning}

\usepackage{cite}
\usepackage[hidelinks]{hyperref}

\title{Quantifying and Minimizing Perception Gap in Social Networks}

\author{} 

\begin{document}

\maketitle
\vspace{-3em}

\begin{center}
\renewcommand{\arraystretch}{1.15}
\begin{tabular}{>{\centering\arraybackslash}p{0.46\textwidth} >{\centering\arraybackslash}p{0.46\textwidth}}
Hemant Kumar Gehlot\footnotemark[2] &
Mohammad Shirzadi\footnotemark[2] \\
Department of Mathematics &
School of Computing \\
Indian Institute of Technology, Roorkee, India &
Australian National University, Canberra, Australia \\
\texttt{hemant\_kg@ma.iitr.ac.in} &
\texttt{mohammad.shirzadi@anu.edu.au} \\[1.4em]
Junhao Gan &
Ahad N. Zehmakan\footnotemark[1] \\
School of Computing and Information Systems (CIS) &
School of Computing \\
The University of Melbourne, Melbourne, Australia &
Australian National University, Canberra, Australia \\
\texttt{junhao.gan@unimelb.edu.au} &
\texttt{ahadn.zehmakan@anu.edu.au} \\
\end{tabular}
\end{center}

\footnotetext[1]{Corresponding author}
\footnotetext[2]{Equal contribution}

\vspace{1.5em}

\begin{abstract}
Social media has transformed global communication, yet its network structure can systematically distort perceptions through effects like the majority illusion and echo chambers. Diagnosing and mitigating these phenomena is a core challenge in graph data mining. We introduce the perception gap index, a graph-based measure that quantifies local-global opinion divergence, which can be viewed as a generalization of the majority illusion to continuous settings. Using techniques from spectral graph theory, we demonstrate that higher connectivity makes networks more resilient to perception distortion. Our analysis of stochastic block models, however, shows that pronounced community structure increases vulnerability. We also study the problem of minimizing the perception gap via link recommendation with a fixed budget. We prove that this problem does not admit a polynomial-time algorithm for any bounded approximation ratio, unless P = NP. Driven by the need for practical tools, we propose a collection of efficient and scalable heuristic algorithms that are demonstrated to produce near-optimal solutions on real-world network data. 
\end{abstract}


\section{Introduction}
In the last few decades, online social networks have expanded rapidly, driven by major platforms that have reshaped global communication. These networks have made it easier for people around the world to share information and stay connected. Although benefits, such as increased access to information and enhanced online interaction, are widely recognized, growing evidence highlights some negative impacts, such as the formation of filter bubbles~\cite{chitra2020analyzing}, echo chambers~\cite{pariser2011filter}, and the majority illusion~\cite{lerman2016majority}, where minority views appear mainstream due to network topology.

There has been a growing body of work in AI sub-communities such as computational social science and graph data mining to understand and control such phenomena in social networks; cf.~\cite{shirzadi2025opinion}. The goal is to quantify measures like polarization~\cite{musco2018minimizing} and majority illusion~\cite{lerman2016majority}, using mathematical concepts such as graph parameters. This then forms the foundation for the study of optimization problems, motivated by real-world applications, such as minimizing polarization through link recommendation using complexity and algorithmic techniques, cf.~\cite{musco2018minimizing}. This line of research sits at the intersection of graph mining and social computing, aiming to develop practical algorithms for improving network safety.

An important phenomenon is how local observations can lead to the misrepresentation of global reality in social networks. Great examples of this phenomenon include the friendship paradox~\cite{Feld1991,eom2014generalized} (where most people believe that they have fewer friends than their peers), the majority illusion~\cite{lerman2016majority}, and echo chambers~\cite{pariser2011filter}. Furthermore, such distortions are believed to be further amplified by recommendation systems that prioritize engagement, inadvertently isolating communities~\cite{pariser2011filter,bhalla2023local}. From a graph data-mining perspective, these phenomena represent systemic distortions that can be diagnosed and potentially mitigated through strategic network interventions. 

A critical aspect of these phenomena is how local observations within a user's social circle can lead to a systematic misperception of global reality. While existing work has effectively modeled binary cases (e.g., the majority illusion), many real-world opinions exist on a spectrum. Furthermore, there is a lack of scalable algorithmic frameworks to actively mitigate such distortions. To bridge this gap, we propose the Perception Gap Index (PGI), a graph-based measure that generalizes the majority illusion to a continuous setting where opinions lie in $[-1, +1]$. This index quantifies, for a given network and opinion distribution, the average discrepancy between a user's local viewpoint and the global average (please see Definition~\ref{def_polarization} for a formal formulation). This can be particularly an indicator of the presence of echo chambers, where users have like-minded connections and a distorted perception of global opinion distribution.

We aim to shed some light on the connection between the structure of the underlying network and the perception gap. Establishing such a connection, then, allows us to propose link recommendation algorithms to reduce the perception gap by connecting strategically selected users (nodes).

\textbf{Network Structure and Perception Gap.} We first study the connection between the perception gap and various graph properties. One of our main findings is that graphs with better connectivity (characterized in algebraic terms) facilitate smaller values of the maximum achievable perception gap (over all choices of opinions). That is, a higher level of connectivity makes a network more ``resilient'' against an adversary who aims to increase the perception gap. Our data-driven observations confirm these theoretical insights on synthetic networks and reveal that community structure is a key vulnerability factor in real-world networks. We also provide further results on special graph structures and opinion distributions, such as the stochastic block model (designed to resemble real-world networks~\cite{abbe2018community}).

\textbf{Inapproximability Results.} Once we have an accurate formulation and a better understanding of the perception gap, one natural question arises: how can one reduce the perception gap in social networks? One popular approach is to add new edges to the graph (inspired by link recommendation algorithms on online platforms). This formulation directly connects to practical link recommendation tasks in social data mining, cf.~\cite{banerjee2023mitigating,wu2024targeted,liu2023fast}. Thus, we study the problem of minimizing the perception gap by adding $k$ new edges for a given network and opinion distribution. We prove, unless P = NP, that there is no polynomial-time algorithm with a bounded approximation ratio for this problem via a novel reduction from the generalized partition problem. Furthermore, we demonstrate that the objective function is neither monotone nor supermodular, underscoring the complexity of minimizing the perception gap.

\textbf{Algorithms and Experiments}. Despite the hardness result, to tackle this problem, we employ several polynomial-time heuristic algorithms, which greedily add edges in a deterministic or randomized fashion. While these algorithms lack theoretical guarantees due to our hardness results, our comprehensive evaluation demonstrates they achieve near-optimal performance on real-world networks (such as Reddit, Twitter, Twitch, Facebook, and LastFM), providing practical tools. To test the accuracy of these algorithms, one needs to compute the optimal solution, which is computationally expensive even for very small networks using a trivial brute-force search. Thus, as a by-product, we develop a non-trivial exact algorithm, which is, of course, not polynomial-time due to the NP-hardness of the problem, but permits us to cover larger graphs.

\subsection{Preliminaries}

\textbf{Graph Notations.} Let $G = (V, E)$ be an undirected, unweighted graph representing a social network, where $V = \{ v_1, \dots, v_n \}$ is the set of nodes and $E$ is the set of edges connecting pairs of nodes. Let $n = |V|$ denote the number of nodes in the graph. For each node $v_i \in V$, the neighborhood of $v_i$, denoted by $N(v_i)$, is the set of nodes $u \in V$ such that $\{ v_i, u \} \in E$. Additionally, for each node $v_i \in V$, let $N[v_i] = N(v_i) \cup \{v_i\}$ denotes the closed neighborhood of $v_i$. The degree of node $v_i$, denoted by $d_i$, is the number of nodes in its \textit{closed} neighborhood, i.e., $d_i = |N[v_i]|$. 

The adjacency matrix $\textbf{A}$ is an $n \times n$ matrix where $\textbf{A}_{ij} = 1$ if there is an edge between nodes $v_i$ and $v_j$ or $i = j$ (conceptually enforcing self-loops), and $\textbf{A}_{ij} = 0$ otherwise. Note that the self-loops enforced in $\textbf{A}$ are just conceptual and do not necessarily exist physically, as we will see shortly.
Alternatively, we could drop this assumption and reformulate the polarization definition accordingly, but the current setup allows better matrix formulation, cosmetically. The degree matrix $\textbf{D} \in \mathbb{R}^{n \times n}$ is a diagonal matrix where $\textbf{D}_{ii} = d_i$
and the Laplacian matrix $\textbf{L}$ of the graph is defined as $\textbf{L} = \textbf{D} - \textbf{A}$. The normalized adjacency matrix $\mathcal{A}$ is defined as $
\mathcal{A} = \textbf{D}^{-\frac{1}{2}} \textbf{A} \textbf{D}^{-\frac{1}{2}}$. Similarly, the normalized Laplacian matrix $\mathcal{L}$ is defined as $\mathcal{L} = \textbf{D}^{-\frac{1}{2}} \textbf{L} \textbf{D}^{-\frac{1}{2}} = \textbf{I} - \mathcal{A}$, where $\textbf{I} \in \mathbb{R}^{n \times n}$ is the identity matrix.

\textbf{Opinion Vector.} For any node $ v_i \in V $, we assign an opinion $s_i \in [-1, 1]$ and let the opinion vector to be defined as $\textbf{S} = [s_1, s_2, \ldots, s_n]^T$, where $-1$ indicates an extreme negative opinion and $+1$ denotes an extreme positive opinion.
For any node $v_i \in V$, we define $\hat{s}_i = \sum_{v_k \in N[v_i]} s_k$ as the sum of opinions in the closed neighborhood of $v_i$ (interpreted as the sum of all opinions that an individual sees, including their own opinion). Let also $\bar{\textbf{S}} = \textbf{S}- \frac{\textbf{S}^T \Vec{\textbf{1}}}{n} \Vec{\textbf{1}}$ be the \textit{mean-centralized opinion vector} across all nodes. Here $\Vec{\textbf{1}}$ is a column vector of length $n$ with all entries as one. 

\textbf{Linear Algebra Notations.} A symmetric matrix $\textbf{M}$ 
is positive definite if $\textbf{x}^T\textbf{M}\textbf{x} > 0$ for all non-zero vectors $\textbf{x}$. If equality is allowed, i.e., $\textbf{x}^T\textbf{M}\textbf{x} \geq 0$ for all non-zero vectors $\textbf{x}$, we call $\textbf{M}$ a positive semi-definite matrix.
For a square matrix $\textbf{M}$, scalar $\lambda$ is an eigenvalue of $\textbf{M}$ if there exists a non-zero vector $\textbf{v}$ such that $\textbf{M}\textbf{v} = \lambda \textbf{v}$. The vector $\textbf{v}$ is called an eigenvector corresponding to $\lambda$, together forming $(\lambda,\textbf{v}) $ an eigenpair of the matrix $\textbf{M}$. Given the symmetry of $\mathbf{M}$, we can express it as $\mathbf{M} = \mathbf{Q} \boldsymbol{\Lambda} \mathbf{Q}^\top$, where $\mathbf{Q}$ is an orthonormal matrix with columns $\{q_1, q_2, q_3, \cdots, q_n\}$, and $\boldsymbol{\Lambda} = \operatorname{\diag}(\lambda_1, \lambda_2, \lambda_3, \cdots, \lambda_n)$ with $\lambda_1 \geq \lambda_2 \geq \cdots \geq \lambda_n$ where $\{(\lambda_i,q_i)\}_{i=1}^n$ are the eigenpair of matrix $\textbf{M}$, for more details see e.g.,~\cite{horn2012matrix,davydov2021meshless}. 
For a non-square matrix $\mathbf{M} \in \mathbb{R}^{m \times n}$, its singular value decomposition is $\mathbf{M} = \mathbf{U} \mathbf{W} \mathbf{V}^T$, where \( \mathbf{U} \in \mathbb{R}^{m \times m} \) and \( \mathbf{V} \in \mathbb{R}^{n \times n} \) are orthogonal matrices whose columns are the eigenvectors of \( \mathbf{M} \mathbf{M}^T \) and \( \mathbf{M}^T \mathbf{M} \), respectively, and $\mathbf{W} \in \mathbb{R}^{m \times n}$ contains singular values $\sigma_1 \geq \sigma_2 \geq \cdots \geq \sigma_r \geq 0$ (\(r = \min\{m, n\}\)) along its diagonal. These singular values are the square roots of the non-zero shared eigenvalues of $\mathbf{M}^T \mathbf{M}$ or $\mathbf{M} \mathbf{M}^T$~\cite{horn2012matrix}.

\textbf{Graph Conductance}. For a graph $G = (V, E)$, let $\delta(S)$ denote the set of edges with exactly one endpoint in $S \subset V$. The volume of $S$, denoted as $\operatorname{vol}(S)$, is the sum of degrees of all nodes in $S$, i.e., $\operatorname{vol}(S) := \sum_{i \in S} d_i$. The conductance of a non-empty proper subset $\emptyset \neq S \subsetneq V$ is defined as
\begin{equation}
\phi(S) := \frac{|\delta(S)|}{\min\{\operatorname{vol}(S), \operatorname{vol}(V \setminus S)\}}.
\end{equation}
The conductance of the graph $G$ is $\phi(G) := \min_{\emptyset \neq S \subsetneq V} \phi(S)$. Let $\mu_i(\mathcal{L})$ denote the eigenvalues of the normalized Laplacian $\mathcal{L} = I - \mathcal{A}$ in non-decreasing order ($0 = \mu_1 \le \mu_2 \le \dots \le \mu_n \le 2$), and $\lambda_i(\mathcal{A})$ denote the eigenvalues of the normalized adjacency matrix $\mathcal{A}$ in non-increasing order ($1 = \lambda_1 \ge \lambda_2 \ge \dots \ge \lambda_n \ge -1$). The second-smallest eigenvalue of $\mathcal{L}$ defines the spectral gap of $G$, satisfying $\mu_2(\mathcal{L}) = 1 - \lambda_2(\mathcal{A})$. By discrete Cheeger's Inequality~\cite{chung1997spectral}, the spectral gap tightly bounds the graph conductance via
\begin{equation}
\frac{\mu_2(\mathcal{L})}{2} \le \phi(G) \le \sqrt{2 \mu_2(\mathcal{L})}.
\end{equation}
\subsection{Quantifying Perception Gap}
We focus on a definition of the perception gap that measures how much people's average views within their social groups differ from the average opinion of the entire network. Formally: 

\begin{definition}[Perception Gap]\label{def_polarization}
For a social network $G = (V, E)$ with $n$ nodes and an opinion vector $\textbf{S} = [s_1, s_2, \cdots, s_n]^T$, we define the perception gap as
\begin{equation}\label{polarization_first_definition}
\P(G, \textbf{S}) = \sum_{i=1}^{n} \left( \frac{\hat{s}_i}{d_i} - \frac{\sum_{j=1}^{n} s_j}{n} \right)^2, 
\end{equation}
where $d_i$ is the degree of node $v_i$ and $\hat{s}_i = \sum_{v_k \in N[v_i]} s_k$ is the sum of opinions in the closed neighborhood of $v_i$.
\end{definition}

Roughly speaking, nodes exposed to a biased opinion contribute more to the perception gap in the network than nodes whose social group is a relatively unbiased sample of the global opinion distribution.

\section{Related Work}

Below, we review key findings on related concepts such as the majority illusion and the echo chamber, as well as similar intervention-optimization problems.

\subsection{Majority Illusion and Structural Bias}
Network structure can distort local perceptions, leading to phenomena such as the friendship paradox and the paradox of majority illusion. The friendship paradox~\cite{lerman2016majority} arises because high-degree nodes are overrepresented in neighborhoods, making most people’s friends seem more connected. This extends to other degree-correlated traits like sentiment~\cite{SentimentParadox2020}. The majority illusion~\cite{lerman2016majority} shows that rare global traits can appear common locally, especially in networks with heterogeneous degrees and disassortative mixing~\cite{newman2002assortative}. Empirical evidence by Lerman et al.~\cite{lerman2016majority} confirms that skewed degree distributions, negative assortativity, and correlations between degree and the attribute amplify the illusion. Also, recent work by Zhang and Peel~\cite{Zhang2023} demonstrates that the network structure itself induces a perception bias, in which sparse connections amplify divergences between local/global opinions.

Grandi et al.~\cite{grandi2022complexity}~initiate the algorithmic study of verifying and eliminating majority illusion. Given a social network and a parameter $q \in [0,1]$ representing the fraction of nodes under illusion, they show that deciding whether a labeling induces a $q$-majority illusion is NP-complete for all rational $q \in \left(\frac{1}{2}, 1\right]$, even on planar, bipartite, and bounded-degree graphs. Moreover, Fioravantes et al.~\cite{fioravantes2025eliminating} prove that this problem is NP-hard, as well as W[2]-hard when parameterized by the solution size, even if the input graph is bipartite and has bounded diameters. In contrast, Dippel et al.~\cite{dippel2025eliminating} prove that complete elimination is solvable in polynomial time.

\subsection{Polarization Dynamics}  

The Friedkin-Johnsen (FJ) model~\cite{Friedkin1990} (see~\cite{shirzadi2025opinion} for a comprehensive review of this model) is widely used to study polarization (variance in opinion distribution at equilibrium). Under this model, R{\'a}cz and Rigobon~\cite{racz2023towards} link graph structural properties to polarization dynamics. Matakos et al.~\cite{matakos2017measuring} study the minimization of polarization by neutralizing $k$ nodes and prove NP-hardness. Musco et al.~\cite{musco2018minimizing} propose a convex, mean-centered polarization objective, solvable under budget constraints. Several follow-up papers have studied link recommendation algorithms aimed at minimizing this polarization (and other adversarial interventions) objective; see, e.g.,~\cite{zhu2021minimizing,zhou2024friedkin}. Chitra and Musco~\cite{chitra2020analyzing} model filter bubbles, showing that small edge weight changes can significantly increase polarization. Bhalla et al.~\cite{bhalla2023local} showed that friend-of-friend links amplify polarization via confirmation bias, while random edge recommendation reduces it.

\subsection{Algorithmic Interventions}
Algorithmic interventions in social networks are commonly leveraged to address structural biases and achieve objectives such as reducing polarization, removing echo chambers, and correcting systematic misperceptions. Common strategies include manipulating node attributes (e.g., opinions~\cite{matakos2017measuring,he2021dynamic}, resistance~\cite{zhou2025efficient}, stubbornness~\cite{shirzadi2024stubborn}) or modifying the network structure itself, see e.g.~\cite{racz2023towards,ren2020detecting}.
Inspired by link recommendation algorithms, the edge addition operation has been widely studied due to its practicality and has been applied to minimize polarization~\cite{garimella2017reducing, gaitonde2020polarization,zhu2021minimizing}, and eliminate the majority illusion~\cite{grandi2022majority, fioravantes2025eliminating}. We similarly study the problem of edge addition to minimize the PGI.

\section{Theoretical Foundation} \label{theory_section2}
In this section, we explore the theoretical underpinnings of the perception gap definition. We provide insights into the connection between graph structural properties and the PGI.

\subsection{Foundation}
For facilitating our analysis, we reformulate Equation~\eqref{polarization_first_definition} in the following abstract matrix form
\begin{equation*}
\P(G, \mathbf{S}) = \left\lVert \left(\mathbf{D}^{-1} \mathbf{A} - \frac{\Vec{\textbf{1}} \Vec{\textbf{1}}^T}{n}\right) \mathbf{S} \right\rVert_2^2. 
\end{equation*}
By the fact that $\mathbf{A} \Vec{\textbf{1}} = \mathbf{D} \Vec{\textbf{1}}$, we have $\P(G, \mathbf{S})= \left\lVert \mathbf{D}^{-1} \mathbf{A}\left(\mathbf{I} - \frac{\Vec{\textbf{1}} \Vec{\textbf{1}}^T}{n}\right) \mathbf{S} \right\rVert_2^2=\left\lVert \mathbf{D}^{-1} \mathbf{A} \mathbf{C_n} \mathbf{S} \right\rVert_2^2$ 
where $\mathbf{C_n} := \mathbf{I} - \frac{\Vec{\textbf{1}} \Vec{\textbf{1}}^T}{n}$ is the centering matrix. Recalling that $\bar{\textbf{S}}=\textbf{S}-\frac{\textbf{S}^T \Vec{\textbf{1}}}{n} \Vec{\textbf{1}}$, we have $\mathbf{C_n} \mathbf{S} = \mathbf{\bar{S}}$, and so the following final form:
\begin{equation} \label{matrix_fom}
\P(G, \mathbf{S}) = \left\lVert \mathbf{D}^{-1} \mathbf{A} \mathbf{\bar{S}} \right\rVert_2^2 = \mathbf{\bar{S}}^T \mathbf{A}^{T} \mathbf{D}^{-2} \mathbf{A} \mathbf{\bar{S}}.
\end{equation}

This representation expresses the perception gap in terms of network structure ($\mathbf{A}$ and $\mathbf{D}$) and mean-centered opinion distribution ($\mathbf{\bar{S}}$). A first observation is that the perception gap remains invariant under centralization of the opinion vector, i.e., $\P(G, \bar{\textbf{S}}) = \P(G, \textbf{S})$. To see this,  observe that $\mathbf{C_n} \mathbf{S} = \mathbf{\bar{S}}$ and $\bar{\textbf{S}}^T \vec{\textbf{1}}  = \textbf{0} $, then we have 
$\mathbf{C_n} \mathbf{C_n} \mathbf{S} = \mathbf{C_n} \mathbf{\bar{S}} = \bar{\textbf{S}} - \frac{\bar{\textbf{S}}^T \vec{\textbf{1}}}{n} \vec{\textbf{1}} = \bar{\textbf{S}} = \mathbf{C_n} \mathbf{S}$.

\subsection{Maximum Achievable Perception Gap}

Suppose, for a graph $G$, that one party (e.g., an adversary) can select an opinion vector to maximize the perception gap. What are the graphs that are resilient against such attacks? That is, what are the graphs for which the maximum achievable perception gap is small? In this section, if $\textbf{S}$ is not given, it is assumed to lie in the domain $[ -1, 1]^n$.

\begin{theorem}[Perception Gap Bounds]\label{Polarization_Bounds}
For any graph $G$ and any opinion vector \textbf{S} where $\| \textbf{S} \|_2 \leq R$, the following bounds hold
\begin{equation*}
\sigma^2_{2}(\mathbf{D}^{-1} \mathbf{A}) \| \mathbf{{S}} \|_2^2  \leq \max_{\mathbf{S},  \lVert \mathbf{S} \rVert_2 \leq R} \P(G, \mathbf{S}) \leq  \sigma^2_{1}(\mathbf{D}^{-1} \mathbf{A}) \| \mathbf{{S}} \|_2^2
\end{equation*}
where $\sigma_1$ and $\sigma_2$ are the largest and second-largest singular values of $\mathbf{D^{-1} A}$, respectively.
\end{theorem}
\begin{proof}
Using the Courant-Fischer theorem~\cite{horn2012matrix}, $\max_{\mathbf{S}, \lVert \mathbf{S} \rVert_2 \leq R} \P(G, \mathbf{S})$ can be written as 
\begin{equation*}
\max_{\mathbf{S}, \lVert \mathbf{S} \rVert_2 \leq R} \mathbf{\bar{S}^T} \mathbf{A} \mathbf{D}^{-2} \mathbf{A} \mathbf{\bar{S} }=\max_{\mathbf{S}, \lVert \mathbf{S} \rVert_2 \leq R, \mathbf{S} \perp \vec{\mathbf{1}}} \mathbf{S}^{T} \mathbf{A} \mathbf{D}^{-2} \mathbf{A} \mathbf{S}.
\end{equation*}
Now, let $\lambda_1 \geq \lambda_2 \geq \cdots \geq \lambda_{n}$ be the eigenvalues of $\mathbf{A} \mathbf{D}^{-2} \mathbf{A}$. Using the variational characterization of eigenvalues
\begin{align*}
\lambda_{2}(\mathbf{A} \mathbf{D}^{-2} \mathbf{A}) \mathbf{S^T} \mathbf{S} 
\leq \max_{\mathbf{S}, \lVert \mathbf{S} \rVert_2 \leq R} \P(G, \mathbf{S}) \leq \lambda_{1}(\mathbf{A} \mathbf{D}^{-2} \mathbf{A}) \mathbf{S^T} \mathbf{S}.
\end{align*}
The proof is done as $\lambda_{i}(\mathbf{A} \mathbf{D}^{-2} \mathbf{A})=\sigma_{i}^2(\mathbf{D}^{-1} \mathbf{A})$.
\end{proof}

We now connect these singular-value bounds to the graph's connectivity. 
By Theorem~\ref{Polarization_Bounds}, the maximum achievable perception gap 
is controlled by the singular values of the random-walk matrix 
$\mathbf{M} := \mathbf{D}^{-1}\mathbf{A}$. To link this to network structure, 
we relate $\sigma_i(\mathbf{M})$ to the spectral gap of the symmetric normalized 
adjacency matrix $\mathcal{A} = \mathbf{D}^{-1/2}\mathbf{A}\mathbf{D}^{-1/2}$, 
which Cheeger's inequality ties to graph conductance.
The bridge is a similarity relation: $\mathcal{A} = \mathbf{D}^{1/2}\mathbf{M}\mathbf{D}^{-1/2}$, 
so $\mathbf{M}$ and $\mathcal{A}$ share the same spectrum. Since $\mathcal{A}$ 
is symmetric, its eigenvalues are real and lie in $[-1,1]$; hence so do those 
of $\mathbf{M}$, and both satisfy
\begin{equation*}
1 = \lambda_1(\mathbf{M}) = \lambda_1(\mathcal{A}) \ge \lambda_2(\mathcal{A}) \ge \dots \ge \lambda_n(\mathcal{A}) \ge -1.
\end{equation*}
Note, however, that similarity does not preserve singular values: 
$\sigma_i(\mathbf{M}) \neq \sigma_i(\mathcal{A})$ in general. In particular, 
$\mathbf{M}$ is row-stochastic ($\mathbf{M}\vec{\mathbf{1}} = \vec{\mathbf{1}}$), 
but $\vec{\mathbf{1}}$ is a singular vector of $\mathbf{M}$ only when $\mathbf{M}$ 
is also column-stochastic, which happens if and only if $G$ is 
degree-regular\footnote{If $\sigma_1(\mathbf{M}) = 1$, then $\vec{\mathbf{1}}$ is 
a right singular vector, so $\mathbf{M}^T\vec{\mathbf{1}} = \vec{\mathbf{1}}$; hence 
$\mathbf{M}$ is doubly stochastic, forcing $d_i = (2|E|+n)/n$ for all $i$. Conversely, 
if $G$ is $d$-regular, then $\mathbf{M} = \frac{1}{d}\mathbf{A}$ is symmetric, so its 
singular values equal the absolute values of its eigenvalues, giving 
$\sigma_1(\mathbf{M}) = 1$.}. Thus $\sigma_1(\mathbf{M}) \ge 1$ always, 
with equality only for regular graphs.
Now applying the additive Weyl inequality~\cite{horn2012matrix} to $\mathbf{M}$ with 
$k = 2$, and using $\lambda_1(\mathbf{M}) = 1$, we obtain 
$1 + |\lambda_2(\mathbf{M})| \le \sigma_1(\mathbf{M}) + \sigma_2(\mathbf{M})$, or 
equivalently
\begin{equation}\label{lambda_2_sigma_2}
\lambda_2(\mathcal{A}) = \lambda_2(\mathbf{M}) \le \sigma_1(\mathbf{M}) + \sigma_2(\mathbf{M}) - 1.
\end{equation}
Equation~\eqref{lambda_2_sigma_2} shows that a smaller sum $\sigma_1(\mathbf{M}) + \sigma_2(\mathbf{M})$ 
forces $\lambda_2(\mathcal{A})$ to be small, which by Cheeger's inequality yields 
a larger spectral gap $1 - \lambda_2(\mathcal{A})$ and hence higher graph conductance. 
Combined with Theorem~\ref{Polarization_Bounds}, this means the same singular values that constrain the perception gap also govern connectivity; the two are structurally linked. For $d$-regular graphs, $\sigma_1(\mathbf{M}) = 1$ and 
the bound tightens to $\lambda_2(\mathcal{A}) \le \sigma_2(\mathbf{M})$, formalized next in Theorem~\ref{maximum_achievable_polarization_d_regular}. Thus, we have the following insights:
\begin{graybox}
\centering
Better-connected, more balanced networks are more resilient against the perception gap.
\end{graybox}

\begin{theorem}\label{maximum_achievable_polarization_d_regular}
For any $d$-regular graph $G$ and any opinion vector $\mathbf{S}$ where $\|\mathbf{S}\|_2 \leq R$, the maximum achievable perception gap is bounded by
\textcolor{black}{$(\sigma_2(\mathcal{A}))^2 R^2$, where $\sigma_2(\mathcal{A})$ is the second-largest singular value of the normalized adjacency matrix $\mathcal{A}$.}
\end{theorem}

\begin{proof}
As \( (\mathbf{D}^{-1} \mathbf{A})^2 \) is a symmetric positive semi-definite matrix, it can be written as 
$
(\mathbf{D}^{-1} \mathbf{A})^2 = \sum_{j=1}^n \lambda_j q_j q_j^T,
$
where \((\lambda_j, q_j)\) is the \(j\)-th eigenpair of \( (\mathbf{D}^{-1} \mathbf{A})^2 \) with \( 1 = \lambda_1 \geq \lambda_2 \geq \cdots \geq \lambda_n \). Since the set of eigenvectors forms a basis for \( \mathbb{R}^n \) and \( \bar{\textbf{S}} \perp q_1 (= \Vec{\textbf{1}}) \), we have \( \bar{\textbf{S}} = \sum_{j=2}^n \alpha_j q_j \).
Using the variational characterization of eigenvalues, we obtain 
$
\max_{\substack{\mathbf{S}, \lVert \mathbf{S} \rVert_2 \leq R}} \mathcal{P}(G, \mathbf{S}) = \max_{\substack{\mathbf{S}, \lVert \mathbf{S} \rVert_2 \leq R}} \sum_{j=2}^n \lambda_j \alpha_j^2 \leq \lambda_2\big((\mathbf{D}^{-1} \mathbf{A})^2\big) \mathbf{S}^T \mathbf{S} =
\textcolor{black}{
(\sigma_2(\mathbf{D}^{-1}\mathbf{A}))^2 R^2 = (\sigma_2(\mathcal{A}))^2 R^2,}
$
\textcolor{black}{where the last two equalities use that for the $d$-regular graph $G$, $\mathbf{M} = \mathbf{D}^{-1}\mathbf{A}$ is symmetric, so $\lambda_i(\mathbf{M}^2) = \sigma_i(\mathbf{M})^2$ when sorted, and $\sigma_i(\mathbf{M}) = \sigma_i(\mathcal{A})$.}
\end{proof}

It is worth mentioning that this bound achieves the value $0$ for the complete graph $K_n$ since 
\textcolor{black}{$\sigma_2(\mathcal{A}_{K_n})=0$} and it is easy to check that $\mathcal{P}(K_n, \textbf{S})=0$, regardless of the opinion vector $\textbf{S}$.

Our results reveal two key insights about the perception gap. (1) Structural origin: Unlike variance-based polarization~\cite{musco2018minimizing}, the perception gap stems from spectral properties ($\sigma_2(\mathbf{D}^{-1}\mathbf{A})$), echoing~\cite{lerman2016majority}'s observation that degree heterogeneity drives local-global discrepancies. (2) Expander resilience: Theorem ~\ref{maximum_achievable_polarization_d_regular}
supports~\cite{hoory2006expander}'s claim that expanders resist distortion. Complete graphs ($\sigma_2=0$) yield $\mathcal{P}=0$, demonstrating perfect resilience, previously seen in consensus dynamics~\cite{olfati2007consensus} but not proved.
\subsection{Random Distribution}
Now, we consider the setup where each opinion $s_i$  follows a standard Gaussian distribution with mean zero and variance ($\sigma_{s_i}^2) =1$. This is a popular and realistic set-up often considered by prior work, cf.~\cite{wang2024relationship}. In the following, we provide a formula for the expected value of the perception gap in a graph $G$ in terms of its degree distribution.
\begin{theorem}\label{expected_value}
For a graph $ G = (V, E) $ and a random opinion vector $ \textbf{S} $ following a standard Gaussian distribution,  the expected value of the perception gap is given by 
\begin{equation}\mathbb{E}\bigg(\P(G,\textbf{S})\bigg)=\bigg(\sum_{v_i \in V} \frac{1}{\deg(v_i)}\bigg) - 1. 
\end{equation}
where $\deg(v_i):=d_i$ denotes the degree of node $v_i$. 
\end{theorem}

\begin{proof}
Let's define a random variable $X_i$ as
\begin{equation*}
    X_i := \left( \frac{\sum_{v_j \in N[v_i]} s_j}{d_i} - \frac{\sum_{j=1}^{n} s_j}{n} \right)^2
\end{equation*}
This can then be simplified as 
\begin{align*}
    X_i &= \underbrace{\frac{1}{d_i^{2}} \bigg(\sum_{v_j \in N[v_i]} s_j\bigg)^2 +\frac{1}{n^2} \bigg(\sum_{j=1}^{n} s_j\bigg)^2}_{\mathcal{I}} -\underbrace{\frac{2}{d_i n}  \bigg( \sum_{v_j \in N[v_i]} s_j \bigg) \bigg( \sum_{j=1}^{n} s_j \bigg)}_{\mathcal{II}}.
\end{align*}
The first term can be rewritten as 
\begin{align*}
    \mathcal{I} &:= \frac{1}{d_i^2} \left( \sum_{v_j \in N[v_i]} s_j^2 + 2 \sum_{\substack{v_j,v_k \in N[v_i] \\ j > k}} s_j s_k \right)  +\frac{1}{n^2} \left( \sum_{j=1}^{n} s_j^2 + 2 \sum_{\substack{j,k=1 \\ j > k}}^n s_j s_k \right),
\end{align*}
and the second term can be reformulated as
\begin{equation*}
    \mathcal{II}:=- \frac{2}{d_i n} \bigg( \sum_{v_j \in N[v_i]} s_j \bigg) \bigg( \sum_{v_j \in N[v_i]} s_j + \sum_{v_j \in V \setminus N[v_i]} s_j \bigg)
\end{equation*}
or 
\begin{align*}
    \mathcal{II} &:= - \frac{2}{d_i n} \Bigg( \sum_{v_j \in N[v_i]} s_j^2 + 2 \sum_{\substack{v_j,v_k \in N(i) \\ j > k}} s_j s_k + \sum_{v_j \in N[v_i]} s_j \sum_{v_j \in V \setminus N[v_i]} s_j \Bigg).
\end{align*}
Since the random variables $ s_i $, for $ i = 1, \dots, n $, are chosen independently and identically distributed (i.i.d.), we have $ \mathbb{E}(s_i s_j) = \mathbb{E}(s_i) \mathbb{E}(s_j) = 0 $ for $i \neq j$ and $\mathbb{E}({s_i}^2) = \text{Var}(s_i) +  (\mathbb{E}({s_i}))^2 = 1 - 0^2 = 1$. Using linearity of expectation, the expected value of the random variable $\mathcal{I}$ will be given as

\begin{equation*}
\mathbb{E}(\mathcal{I}) := \frac{1}{d_i^2}  |N[v_i]| 
    + \frac{1}{n^2} n= \frac{1}{d_i}+\frac{1}{n}, 
\end{equation*}
Similarly, 
\begin{equation*}
\mathbb{E}(\mathcal{II}) := - \frac{2}{d_i n} |N[v_i]|=- \frac{2}{ n}.  
\end{equation*}
So the expected value of $X_i$ is given as $\mathbb{E}(X_i)=\frac{1}{d_i}-\frac{1}{n}$, and the proof follows by summing over all nodes.
\end{proof}

\textbf{Note.} We adopt Gaussian opinions in Theorem~\ref{expected_value} since this assumption is widely used in the literature and allows analytical tractability. However, the result is not restricted to the Gaussian case—it remains valid for any i.i.d. opinion distribution with zero mean and unit variance. Thus, the derived expression also provides useful insights for bounded-opinion scenarios considered in our model.

This theorem implies that increasing the degree of nodes helps reduce the expected perception gap. (This, in particular, returns a value $0$ for a complete graph $K_n$, which is exactly as discussed above.) This is intuitive, since if a node has more neighbors, it's more likely that its local view of opinions matches the global one. On the other hand, 
\begin{graybox}
\begin{center}
Graphs with many low-degree nodes, such as real-world networks following a power-law degree distribution, have a larger expected perception gap.
\end{center}
\end{graybox}

This result highlights the critical role of degree distribution in shaping the perception gap, reinforcing the idea that increasing connectivity, especially for low-degree nodes, can substantially mitigate local-global opinion discrepancies in networks. This aligns with the majority illusion phenomenon~\cite{lerman2016majority}, in which degree heterogeneity leads local perceptions to deviate from the global reality systematically.

\subsection{The Stochastic Block Model Analysis}
\label{sec:sbm}
We now analyze a special case of stochastic block models (SBM)~\cite{abbe2018community}, often used to synthesize real-world networks.
It is defined by $2n$ nodes, and $p, q \in [0,1]$. Nodes are divided into groups $V_+ = \{v_1, \ldots, v_n\}$ and $V_- = \{v_{n+1}, \ldots, v_{2n}\}$, with intra-group edges formed with probability $p$ and inter-group edges with probability $q$ ($q < p$) reflecting the $\textit{homophily}$ principle~\cite{dandekar2013biased}. Nodes in $V_+$ hold opinions $+1$, and those in $V_-$ hold $-1$. 
The expected adjacency matrix $\bar{\textbf{A}}$ is a $2n \times 2n$ block matrix given as follows (where the diagonal entries are $1$ as we reinforce the self-loop). 
\begin{equation}\label{adjacency_A}
\bar{\textbf{A}} = \left[ 
\begin{array}{cccc|cccc}
\cellcolor{cyan!20} 1 & \cellcolor{cyan!20} p & \cellcolor{cyan!20} \cdots & \cellcolor{cyan!20} p & \cellcolor{red!20} q & \cellcolor{red!20} q & \cellcolor{red!20} \cdots & \cellcolor{red!20} q \\
\cellcolor{cyan!20} p & \cellcolor{cyan!20} 1 & \cellcolor{cyan!20} \cdots & \cellcolor{cyan!20} p & \cellcolor{red!20} q & \cellcolor{red!20} q & \cellcolor{red!20} \cdots & \cellcolor{red!20} q \\
\cellcolor{cyan!20} \vdots & \cellcolor{cyan!20} \vdots & \cellcolor{cyan!20} \ddots & \cellcolor{cyan!20} \vdots & \cellcolor{red!20} \vdots & \cellcolor{red!20} \vdots & \cellcolor{red!20} \ddots & \cellcolor{red!20} \vdots \\
\cellcolor{cyan!20} 1 & \cellcolor{cyan!20} p & \cellcolor{cyan!20} \cdots & \cellcolor{cyan!20} 1 & \cellcolor{red!20} q & \cellcolor{red!20} q & \cellcolor{red!20} \cdots & \cellcolor{red!20} q \\
\hline
\cellcolor{red!20} q & \cellcolor{red!20} q & \cellcolor{red!20} \cdots & \cellcolor{red!20} q & \cellcolor{cyan!20} 1 & \cellcolor{cyan!20} p & \cellcolor{cyan!20} \cdots & \cellcolor{cyan!20} p \\
\cellcolor{red!20} q & \cellcolor{red!20} q & \cellcolor{red!20} \cdots & \cellcolor{red!20} q & \cellcolor{cyan!20} p & \cellcolor{cyan!20} 1 & \cellcolor{cyan!20} \cdots & \cellcolor{cyan!20} p \\
\cellcolor{red!20} \vdots & \cellcolor{red!20} \vdots & \cellcolor{red!20} \ddots & \cellcolor{red!20} \vdots & \cellcolor{cyan!20} \vdots & \cellcolor{cyan!20} \vdots & \cellcolor{cyan!20} \ddots & \cellcolor{cyan!20} \vdots \\
\cellcolor{red!20} q & \cellcolor{red!20} q & \cellcolor{red!20} \cdots & \cellcolor{red!20} q & \cellcolor{cyan!20} p & \cellcolor{cyan!20} p & \cellcolor{cyan!20} \cdots & \cellcolor{cyan!20} 1
\end{array} 
\right]_{2n \times 2n}
\end{equation}

The following theorem provides a closed-form expression for the perception gap in an SBM graph. 
\begin{theorem}\label{polarization_SBM}
Let $G = (V, E)$ be an SBM graph with two blocks representing a fully polarized society, where one group holds an opinion of $+1$ and the other $-1$. The perception gap in this case is given by
\begin{equation}
\P(G,\mathbf{S}) = \left( \frac{p - q + \frac{1 - p}{n}}{p + q + \frac{1 - p}{n}} \right)^2 2n.
\end{equation}
\end{theorem}

\begin{proof}
In this setup, we have $\mathbf{\bar{S}} = \mathbf{S}$, and hence $\bar{\mathbf{A}}\mathbf{S} = (n(p-q) + 1 - p)\mathbf{S}$. Consequently the degree matrix is given by $\mathbf{D} = \operatorname{diag}(\bar{\mathbf{A}}\Vec{\mathbf{1}}) = (n(p+q) + 1 - p)\mathbf{I}$. Therefore, 
\begin{align*}
\P(G, \mathbf{S}) 
&= \left\lVert \mathbf{D}^{-1} \mathbf{A} \mathbf{\bar{S}} \right\rVert_2^2 
= \left\lVert \frac{n(p-q)+1-p}{n(p+q)+1-p} \mathbf{S} \right\rVert_2^2 \\
\end{align*}
which completes the proof as $\left\lVert \mathbf{S} \right\rVert_2^2 = 2n$.
\end{proof}

For large $n$, which is typical in real-world networks, the terms $\frac{1-p}{n}$ in both the numerator and denominator of Theorem~\ref{polarization_SBM} are negligible compared to $p$ and $q$. The perception gap thus simplifies as
\[
\P(G, \mathbf{S}) \approx \left( \frac{p-q}{p+q} \right)^2 2n.
\]
This form shows that, for sufficiently large $n$, the perception gap is primarily determined by the relative difference $p-q$, i.e., the gap between the strength of intra-community and inter-community connectivity. So as the communities become more separated (larger $p-q$), local nodes' perception diverges further from the global average.

Recent work analyzing FJ dynamics on SBM graphs~\cite{chitra2020analyzing} shows that standard polarization, the variance of final expressed opinions, is essentially independent of the intra-community edge probability $p$. This is counterintuitive, as increasing $p$ (stronger within-community links) does not alter the polarization value, even though it increases the separation of communities. Our PGI, by contrast, directly captures the effect of both $p$ and $q$. As $p-q$ increases, indicating stronger community structure amplifies opinion discrepancies while standard polarization stays unchanged. Thus, our metric complements variance-based polarization by quantifying the influence of network modularity on opinion divergence, revealing structural effects masked by existing measures. This aligns with prior work on how homophily and modularity distort perceptions and increase polarization~\cite{lerman2016majority, chitra2020analyzing, pariser2011filter}.

\section{
Perception Gap Minimization via Edge Addition
}\label{link_recommendation_section3}

In this section, we study the problem of strategically adding a limited number of edges to minimize the perception gap, thereby mitigating systematic misperceptions and fostering a more accurate understanding of the global truth. The same algorithmic study has been applied to other adversarial phenomena like polarization~\cite{springsteen2024algorithmic,cinus2022effect,garimella2017reducing,auletta2024mitigate,interian2020reducing} and majority illusion~\cite{grandi2022majority, fioravantes2025eliminating}. 

\begin{problem}[\text{Perception Gap Minimization via Edge Addition}]\label{minkpolarization}
Given a graph $G = (V, E)$, an opinion vector $\textbf{S}$,  the set of missing edges $\mathcal{E}_C = \big\{ \{u, v\} : \{u, v\} \notin E \big\}$, and an integer $k $, find a subset $T \subseteq \mathcal{E}_C$ with $|T| = k$, such that adding the edges in $T$ to $G$, forming the new graph $G + T = (V, E \cup T)$, minimizes perception gap, i.e., 

\begin{equation}
    \min_{T \subseteq \mathcal{E}_C, \hspace{2mm} |T| = k} f(T), \quad f(T)=\P(G+T,\textbf{S}).
\end{equation}
\end{problem}

In this section, we show that this problem is computationally hard; more precisely, there is no polynomial-time algorithm that can guarantee a bounded approximation ratio unless P = NP (Theorem~\ref{thm:np-hardness}) and its objective function is neither monotone \footnote{A set function $f: 2^N \to \mathbb{R}$ is \textit{monotonically decreasing} if for all subsets $A \subseteq B \subseteq N$, the function satisfies $f(B) \leq f(A)$. 
} nor supermodular \footnote{
A set function $f: 2^N \to \mathbb{R}$ is \textit{supermodular} if for all subsets  $A \subseteq B \subseteq N$ and for all elements $x \notin B$, $
f(A \cup \{x\}) - f(A) \leq f(B \cup \{x\}) - f(B)$. 
} (Lemma~\ref{monotonicity_submodularity}). Before the main results, we present a lemma giving a closed-form expression for how adding a link alters the perception gap. 
\begin{lemma}\label{Evolution_thorem}
Given an opinion vector $\mathbf{S}$ and a social network $G = (V, E)$ represented by the adjacency matrix $\mathbf{A}$ and degree matrix $\mathbf{D}$, let $G_{+e} = (V, E \cup \{e\})$ denote the updated social network formed by adding a new edge $e = \{v_i, v_j\} \notin E$. The change in perception gap caused by adding edge $e$, i.e, $\P(G_{+e},\textbf{S}) - \P(G,\textbf{S})$
is given by
\begin{align*}
&\left[ \left( \frac{\hat{s}_i + s_{j}}{d_{i}+1} - \bar{s} \right)^2 - \left( \frac{\hat{s}_i}{d_{i}} - \bar{s} \right)^2 \right] + \left[ \left( \frac{\hat{s}_j + s_{i}}{d_{j}+1} - \bar{s} \right)^2 - \left( \frac{\hat{s}_j}{d_{j}} - \bar{s} \right)^2 \right].
\end{align*}
\end{lemma}

\begin{proof}
Using the notation $\bar{s}_k$ and $\bar{s}_k'$ for the local perspective in the original graph, $G$, and augmented graph, $G_{+e}$, of node $v_k$, and following the straightforward calculations show that  
\begin{align*}
\Delta_{+e} \mathbf{P} &= \P(G_{+e},\textbf{S}) - \P(G,\textbf{S}) \\
&= \sum_{k=1}^{n} \left( \frac{\bar{s}_k'}{d_k'} - \bar{s} \right)^2 - \sum_{k=1}^{n} \left( \frac{\bar{s}_k}{d_{k}} - \bar{s} \right)^2 \\
&= \sum_{k=1}^{n} \left[ \left( \frac{\bar{s}_k'}{d_{k}'} - \bar{s} \right)^2 - \left( \frac{\bar{s}_k}{d_{k}} - \bar{s} \right)^2 \right] \\
&= \sum_{k=1}^{n} \left[ \left( \frac{\bar{s}_k'}{d_{k}'} - \bar{s} \right)^2 - \left( \frac{\bar{s}_k}{d_{k}} - \bar{s} \right)^2 \right].
\end{align*}
The proof is completed by following the observations below
\begin{equation*}
\bar{s}_k' = 
\begin{cases}
\bar{s}_k + s_j, &\text{if $k=i$}\\
\bar{s}_k+ s_i, &\text{if $k=j$}\\
\bar{s}_k, & \text{otherwise.} 
\end{cases}, \quad 
d_{k}' = 
\begin{cases}
d_{k} + 1, &\text{if $k=i,j$}\\
d_{k}, & \text{o.w.} 
\end{cases}
\end{equation*}
for $ k = 1,2, \cdots, n $.
\end{proof}
We also provide an example to better understand the implications of this lemma. 
Consider Figure~\ref{example_reduce}, where node labels represent opinion values. (Please note, as mentioned before, although self-loops are not explicitly shown in the graph, they are always present to make $\mathbf{D^{-1}}$ well-defined for all kinds of graphs.)
\begin{figure}
\begin{center}
\begin{tikzpicture}[
    vertex/.style={circle, draw, minimum size=4mm, inner sep=1pt},
    label/.style={font=\scriptsize},
    every edge/.append style={thick}
]
\node[vertex] (A) at (0,0) {$0$};
\node[vertex] (B) at (1.5,0) {$-1$};
\node[vertex] (C) at (3,-1) {$-1$};
\node[vertex] (D) at (3,1) {$-1$};
\node[vertex] (E) at (-1.5,0) {$+1$};
\node[vertex] (F) at (-3,1) {$+1$};
\node[vertex] (G) at (-3,-1) {$+1$};

\node[font=\small\bfseries, text=black] at (0, 1.3) {Initial $\mathcal{P} = 5.1$};

\node[label, above=0.1mm of A] {A};
\node[label, above=0.1mm of B] {B};
\node[label, below right=0.1mm of C] {C};
\node[label, above right=0.1mm of D] {D};
\node[label, above=0.1mm of E] {E};
\node[label, above left=0.1mm of F] {F};
\node[label, below left=0.1mm of G] {G};

\draw (A) -- (B);
\draw (B) -- (C);
\draw (C) -- (D);
\draw (D) -- (B);
\draw (A) -- (E);
\draw (E) -- (F);
\draw (F) -- (G);
\draw (G) -- (E);


\draw[red, dashed, thick] (G) -- (C) 
    node[midway, sloped, fill=white, inner sep=1.5pt, font=\tiny\bfseries, text=red!80!black] 
    {$[\mathcal{P} \to 3.6]$};

\draw[blue, dashed, thick] (E) to[bend right=5] 
    node[midway, sloped, fill=white, inner sep=1.5pt, font=\tiny\bfseries, text=blue!80!black] 
    {$[\mathcal{P} \to 3.9]$} (C);

\draw[orange!90!black, dashed, thick] (A) -- (C) 
    node[midway, sloped, fill=white, inner sep=1.5pt, font=\tiny\bfseries, text=orange!90!black] 
    {$[\mathcal{P} \to 4.7]$};
    
\end{tikzpicture}

\caption{Impact of strategic edge additions on reducing the perception gap ($\mathcal{P}$). Edge $(G, C)$ (red) yields the largest PGI reduction connecting nodes with strongly divergent local perspectives.}
\label{example_reduce}
\end{center}
\end{figure}

The initial perception gap is $5.1$. Here are some key observations: First, adding edge GC (highlighted in red) reduces the perception gap from $5.1$ to $3.6$, marking the most significant decrease in the perception gap for this graph. A plausible explanation for this is that both nodes G and C hold firm, opposing views that align with the opinions of their local neighborhoods. Second, adding edge EC (blue) reduces the perception gap from $5.1$ to $3.9$, which has a moderate effect. Although E and C have strong opposing views, E's local neighborhood is less extreme. Lastly, adding edge AC (orange) reduces the perception gap to $4.7$, which has the least impact on this reduction. This is likely because node A is neutral and doesn't significantly differ from C's opinion.
This example illustrates how the interplay between individual opinions and local neighborhood perceptions influences the network's overall perception gap when new connections are formed. They highlight the potential for targeted interventions (like edge GC in the above example) in network structure to mitigate the perception gap.

Before proving Theorem~\ref{thm:np-hardness}, we need to define the Generalized Partition Problem (GPP), since it serves as the base problem for our reduction for the hardness argument. We begin by formally defining the GPP and recalling its computational complexity.

\begin{problem}{\label{GPP}} \textbf{Generalized Partition Problem (GPP):}
Given a set of $N$ integers $V = \{v_1, v_2, \ldots, v_N\}$ and a value $k$, define a target $b = \frac{k}{N} \cdot \sum_{i=1}^{N} v_i$. The task is to find a subset $B \subseteq V$ of size $k$ that minimizes 
    $\left| \sum_{v_i \in B} v_i - b \right|$. 
\end{problem}
For $k = n/2$, the GPP turns out to be the Partition Problem, which is somewhat similar, in which the goal is to determine whether the given set of integers can be partitioned into two subsets with equal sums. The Partition Problem is NP-hard, cf.~\cite{kovalyov2010generic}. Therefore, any polynomial-time approximation algorithm for GPP with ratio $\alpha > 1$ would imply a polynomial-time algorithm for the Partition Problem (since it would allow us to distinguish between zero and non-zero absolute values, which determines the Yes/No answer to the Partition Problem), which contradicts its NP-hardness. Hence, there is no polynomial-time algorithm with a fixed multiplicative approximation ratio $\alpha > 1$ for solving the GPP, unless P = NP.
We now use a reduction from GPP to establish the following result.

\begin{theorem} \label{thm:np-hardness}
There is no polynomial-time algorithm with a bounded approximation ratio for solving Problem~\ref{minkpolarization} unless P=NP.
\end{theorem}

\begin{proof} 
Consider an instance of the GPP where we are given a set of $N$ integers, $V = \{v_1, v_2, \ldots, v_N\}$ and a value $k$, and the goal is to find a subset $B \subseteq V$ of size $k$ such that for any target value as $b = \frac{k}{N}  \sum_{i=1}^{N} v_i$, $\left|\sum_{v_i \in B} v_i - b\right|$ is minimized. Now, we construct an instance of Problem~\ref{minkpolarization} from the instance of GPP as follows:

Begin with a complete graph $K_N$ with nodes $U = \{u_1, u_2, \ldots, u_N\}$. Let $v_{\max} = \max_{v_i \in V} |v_i|$ and assign an opinion $s_i = \frac{v_i}{v_{\max}}$ to each node $u_i$. Add a node $u_0$ with opinion $s_0 = t$ (the value of $t$ will be determined later). 
Now, form the network $G(U \cup \{u_0\}, E) = K_N \cup \{u_0\}$
. To complete the proof, we must demonstrate that solving Problem~\ref{minkpolarization} on this constructed instance is equivalent to solving the original GPP. Let $W = \sum_{i = 1}^{N} s_i$. Let $T$ be any subset of size $k$ of missing edges of graph $G$ and $b'$ be the sum of opinions of the terminal nodes of edges in $T$ except $u_0 $, i.e., $ b' = \sum_{\{u_i, u_0\} \in T} s_i$.
Figure~\ref{fig:example_construction} illustrates this construction for $N=7$ and $k=3$.

\begin{figure}[h]
    \centering
    \begin{tikzpicture}[scale=1.5]
        \def\N{7}

        \foreach \i in {1,...,\N} {
            \node[draw, circle] (u\i) at ({360/\N * (\i - 1)}:1) {$u_{\i}$};
        }

        \foreach \i in {1,...,\N} {
            \foreach \j in {\i,...,\N} {
                \ifnum\i<\j
                    \draw (u\i) -- (u\j);
                \fi
            }
        }

        \node[draw, circle] (u0) at (4,0) {$u_0$};

        \draw[red, thick] (u0) -- (u1);
        \draw[red, thick] (u0) -- (u2);
        \draw[red, thick] (u0) -- (u6);
    \end{tikzpicture}
\caption{\textcolor{black}{Illustration of the reduction from GPP to Problem~\ref{minkpolarization} ($N=7, k=3$). Red edges represent the candidate edge subset $T = \{\{u_1, u_0\}, \{u_2, u_0\}, \{u_6, u_0\}\}$ connecting $u_0$ to nodes $\{u_1, u_2, u_6\} \subset K_7$. The opinion sum of these connected nodes is $b' = s_1 + s_2 + s_6$, resulting in a modified perception gap of $\mathcal{P}(G(U, E + T), \mathbf{S}) = \left(\frac{(v_1 + v_2 + v_6) - b}{v_{\max}(k + 1)}\right)^2$.}}
\label{fig:example_construction}
\end{figure}

Then, the perception gap of the graph $G$ with the addition of edges in $T$ will be calculated as 
\begin{align*}
\P(G(U, E + T), \textbf{S}) = \left( \frac{t + b'}{k + 1} - \frac{W + t}{N + 1} \right)^2 + (N-k) \left( \frac{W}{N} - \frac{W + t}{N + 1} \right)^2.
\end{align*}
By setting $t = \frac{W}{N}$ we will have 
\begin{align*}
\P(G(U, E + T), \textbf{S}) &= \left( \frac{\frac{W}{N} + b'}{k + 1} - \frac{W + \frac{W}{N}}{N + 1} \right)^2+ (N-k) \left( \frac{W}{N} - \frac{W + \frac{W}{N}}{N + 1} \right)^2.
\end{align*}
which can be simplified as 
\begin{equation*}
\P(G(U, E + T) ,\textbf{S}) =  \left( \frac{b'N - Wk }{N(k + 1)}   \right)^2
\end{equation*}
Now, as $W = \sum_{i = 1}^{N} s_i = \frac{\sum_{i = 1}^{N} v_i}{v_{\max}}$ where $v_{\max} = \max |v_i|$. Therefore, target value $b = \frac{k}{N} \cdot(\sum_{i=1}^{N} v_i ) = \frac{v_{\max} k W}{N}$. Hence, 
\begin{equation} 
\P(G(U, E + T) ,\textbf{S})  = \left( \frac{b' - \frac{b}{v_{\max}} }{k + 1}   \right)^2 = \left( \frac{\sum_{v_i \in B} v_{i} - {b}}{v_{\max}(k + 1)} \right)^2. \label{Reduction_polarisation0}
\end{equation}

Let, on the contrary, some polynomial-time algorithm ($A$) with an approximation ratio, say $\mathbf{\alpha}$ $(> 1)$, for Problem~\ref{minkpolarization} exist. Let $T^{A}$ (with corresponding $B^{A} \subseteq V $) be the subset of missing edges of $G$ obtained by the algorithm with perception gap value as $\P (G +  T^{A}, \textbf{S})$. Let $b^{A}$ be the sum of opinions of the terminal nodes of edges in $T^{A}$ except $u_0$. Let $T^{OPT}$ (with corresponding $B^{OPT} \subseteq V $) be the optimum subset of missing edges of $G$ with optimum value $\P (G+T^{OPT}, \textbf{S}) $ and $b^{OPT}$ be the sum of opinions of the terminal nodes of edges in $T^{OPT}$ except $u_0$.Then we have the following:
\begin{equation*}
    \P(G+T^{A},\textbf{S})  \leq \alpha \P(G+T^{OPT},\textbf{S}). 
\end{equation*}
Using Equation~\eqref{Reduction_polarisation0} we will have 
\begin{equation*}
\left( \frac{b^A - \frac{b}{v_{\max}} }{k + 1}   \right)^2  \leq \alpha \left( \frac{b^{OPT} - \frac{b}{v_{\max}} }{k + 1}   \right)^2
\end{equation*}
which is equivalent to 
\begin{equation*}
\left( \frac{\sum_{v_i \in B^A} v_{i} - {b}}{v_{\max}(k + 1)} \right)^2 \leq \alpha \left( \frac{\sum_{v_i \in B^{OPT}} v_{i} - {b}}{v_{\max}(k + 1)} \right)^2, 
\end{equation*}
or 
\begin{equation*}
\left(\sum_{v_i \in B^A} v_{i} - {b} \right)^2  \leq \alpha \left(\sum_{v_i \in B^{OPT}}v_{i} - {b} \right)^2
\end{equation*}

which gives us $\left|\sum_{v_i \in B^A} v_{i} - {b} \right|  \leq \sqrt{\alpha} \left|\sum_{v_i \in B^{OPT}}v_{i} - {b} \right|$. This means that we have a $\sqrt{\alpha}$ approximation algorithm for the GPP. However, this can't be the case since we argued above. Thus, we reach a contradiction.
\end{proof}

\begin{lemma}\label{monotonicity_submodularity}
The function $f(\cdot)$ in Problem~\ref{minkpolarization}, is neither monotone nor supermodular.
\end{lemma}

\begin{proof}
The function $f(T) = \P(G+T,\textbf{S})$ is not monotonically decreasing. Consider the graph shown in Figure~\ref{fig:non-mono-submod} (a) for this. Here, adding the edge $\mathbf{BD}$ increases the perception gap.
Secondly, the function is not supermodular. Consider the network shown in Figure~\ref{fig:non-mono-submod} (b). We have the network with no edges and define two edge sets, $\mathbf{T} = \{e_3\}$ and $\mathbf{W} = \{e_3, e_2\}$. Straightforward calculation shows that 
    $f(T\cup \{e_1\}) - f(T) = 0.25 > -0.25 =  f(W \cup \{e_1\}) - f(W)$ 
which violates the definition of supermodularity. 
\end{proof}

\begin{figure}[htbp]
\centering
\begin{tikzpicture}[scale=0.85,
    vertex/.style={circle, draw, minimum size=6mm, inner sep=0pt, font=\tiny\bfseries},
    label/.style={font=\scriptsize\bfseries}
]

\begin{scope}[xshift=-4cm]
    \node[font=\small\bfseries] at (0, 1.8) {(a) Non-Monotonicity};
    
    \node[vertex, fill=cyan!40] (A) at (1,1) {$+1$};
    \node[vertex, fill=red!30] (B) at (-1,1) {$-1$};
    \node[vertex, fill=red!30] (C) at (1,-1) {$-1$};
    \node[vertex, fill=cyan!40] (D) at (-1,-1) {$+1$};

    \node[label, right=1pt of A] {A};
    \node[label, left=1pt of B] {B};
    \node[label, right=1pt of C] {C};
    \node[label, left=1pt of D] {D};

    \draw[thick] (A) -- (B);
    \draw[thick] (C) -- (D);

    \draw[red, dashed, thick] (D) -- (B) node[midway, left, font=\bfseries] {$e$};

    \node[draw=gray!50, fill=gray!5, rounded corners=2pt, font=\scriptsize, align=center] at (0, -2.3) {
     $f(\emptyset) = 0 \longrightarrow f(\{e\}) = 0.22$\\
        $f(\{e\}) - f(\emptyset) = +0.22 > 0$
    };
\end{scope}

\begin{scope}[xshift=1.3cm]
    \node[font=\small\bfseries] at (0, 1.8) {(b) Non-Supermodularity};

    \node[vertex, fill=cyan!40] (A) at (1,1) {$+1$};
    \node[vertex, fill=red!30] (B) at (-1,1) {$-1$};
    \node[vertex, fill=yellow!40] (C) at (0,-1) {$0$};

    \node[label, right=1pt of A] {A};
    \node[label, left=1pt of B] {B};
    \node[label, left=1pt of C] {C};

    \draw[blue, dotted, thick] (A) -- (C) node[midway, right, font=\bfseries, text=blue!80!black] {$e_1$};
    \draw[orange!90!black, dashed, thick] (C) -- (B) node[midway, left, font=\bfseries, text=orange!90!black] {$e_2$};
    \draw[black, dashed, thick] (A) -- (B) node[midway, above, font= \bfseries] {$e_3$};

    \node[draw=gray!50, fill=gray!5, rounded corners=2pt, font=\scriptsize, align=center] at (0, -2.3) {
     Let $T = \{e_3\}, W = \{e_3, e_2\}$ ($T \subset W$)\\
    $f(T \cup \{e_1\}) - f(T) = 0.25 > $ \\ 
    $ f(W \cup \{e_1\}) - f(W) = -0.25 $
    };
\end{scope}

\end{tikzpicture}
\caption{\textcolor{black}{Counterexamples illustrating structural hardness of PGI minimization. \textbf{(a) Non-monotonicity}: Edge $(D,B)$ increases PGI. \textbf{(b) Non-supermodularity}: Candidate edge $e_1$ yields a larger marginal PGI reduction in $W$ than in subset $T \subset W$.}}
\label{fig:non-mono-submod}
\end{figure}

\section{Algorithms}\label{algorithms}
As we proved in Theorem~\ref{thm:np-hardness}, Problem~\ref{minkpolarization} does not allow for any polynomial-time algorithm with a bounded approximation guarantee. Thus, we resort to some heuristics: (i) \textbf{Random}: Add edges uniformly at random; (ii) \textbf{Greedy}: Iteratively add the edge that maximally reduces perception gap (Algorithm~\ref{alg:spgreedy} with $b=1$); (iii) \textbf{Batch Greedy}: Iteratively add the top $b$ edges ranked by \textbf{Greedy} (with $b=5, 10$ in our experiments); (iv) \textbf{Random Batch Greedy}: Randomly select one edge from the top $b$ edges ranked by batch greedy.
\begin{algorithm}[ht]
\small
\caption{Greedy($G$, $k$, $b$, $\mathbf{S}$)}
\label{alg:spgreedy}
\begin{algorithmic}[1]
\Require Graph $G=(V,E)$, integer $k \le |\mathcal{E}_C|$, batch size $b$, opinion vector $\mathbf{S}$
\Ensure Subset $T \subseteq \mathcal{E}_C$ with $|T| = k$
\State Initialize solution $T \gets \varnothing$
\While{$|T| < k$}
    \State Compute $f(e)$ for each $e \in \mathcal{E}_C \setminus T$
    \State Sort all missing edges $e \in \mathcal{E}_C \setminus T$ in non-decreasing order based on $f(e)$
    \State Select the top $t = \min\{b, k-|T|\}$ edges, $\{e_1,\dots,e_t\}$ from the sorted list
    \State Update solution: $T \gets T \cup \{e_i\}_{i=1}^t$
    \State Update graph: $G \gets (V, E \cup \{e_i\}_{i=1}^t)$
\EndWhile
\State \Return $T$
\end{algorithmic}
\end{algorithm}

\textcolor{black}{
\textbf{Complexity of Algorithm~\ref{alg:spgreedy}.} Assuming an adjacency-list representation, the initial computation of $d_i$, $\hat{s}_i$, and the global mean takes $O(n + |E|)$ precomputation time. In each batch, evaluating each candidate edge takes $O(1)$ time via Lemma~\ref{Evolution_thorem}, followed by sorting all missing edges in $O(|\mathcal{E}_C| \log |\mathcal{E}_C|)$ time. Updating $d_i$ and $\hat{s}_i$ for the endpoints of the $b$ selected edges takes $O(b)$ time per batch, totaling $O(k)$ across all batches. Thus, across $\lceil k/b \rceil$ batches, the total time complexity of Algorithm~\ref{alg:spgreedy} is $O\big(n + |E| + \lceil k/b \rceil \cdot |\mathcal{E}_C| \log |\mathcal{E}_C| + k\big)$.
}


We aim to compare the results of these algorithms with the optimal (exact) solution. However, since the search space for finding the optimal solution grows exponentially, it becomes computationally prohibitive to evaluate graphs larger than $20-30$ nodes using a brute-force approach. To address this limitation and consider larger graphs, we propose a non-trivial exact algorithm that reduces the search space. This algorithm first sorts all missing edges by their maximum potential contribution (MPC) to reducing the perception gap. MPC of an edge $e= \{u,v\} \in \mathcal{E}_C$ is the maximum reduction achievable by adding $e$ in all possible scenarios when $k' = 0, 1, 2, \dots, k - 1$ edges have already been added to the neighborhood $N(u) \cup N(v)$. Roughly speaking, MPC measures the best reduction one can hope to achieve by adding $e$ when adding $k$ edges. An edge $e$ can then be safely disregarded if the summation of MPC for the edge $e$ and the top $k-1$ edges (ranked by MPC) is less than the reduction achieved by some other solution, for example, the one obtained by the \textbf{Greedy} approach. The reason is that if an edge cannot produce a better solution even when paired with the $k-1$ best edges and their best possible contributions (MPCs) considered, then it cannot be part of an optimal solution. Please see Algorithm~\ref{mpc_alg} for a formal description. The correctness of the MPC algorithm can be justified by analyzing both the calculation of the MPC values and the filtering process for removing ineffective edges.

First, the algorithm correctly computes the MPC value, an upper bound on the potential contribution to reducing the perception gap for each missing edge, when adding $k$ edges. For each edge $e$, the algorithm iterates over all combinations of $k'$ edges from the neighborhood $N(u) \cup N(v)$ and calculates the reduction in perception gap for each combination. This is because we are interested in all possible scenarios for $e$ when we add $k$ edges, and we focus only on incident edges, since they are the only ones that affect the value we are interested in. The edge $e$ is assigned the maximum reduction in the perception gap across all these scenarios, ensuring that the computed MPC value accurately reflects its maximum potential impact on perception gap reduction.

Next, once all edges are ranked by their MPC values in descending order, the algorithm computes $S_e = mpc_e + \sum_{i=1}^{k-1} mpc_{i}$, where $ mpc_e$ is the MPC of edge $e$, and $\sum_{i=1}^{k-1} mpc_{i}$ is the sum of the MPC values of the top $k-1$ edges. This sum represents an upper bound on the cumulative contribution of edge $e$ and the most beneficial edges to reducing the perception gap. Suppose $S_e$ is smaller than the perception gap reduction achieved by any other solution, particularly the one obtained by the \textbf{Greedy} algorithm. In that case, the edge is removed from the list of missing edges. The reason is that if an edge cannot contribute to an optimal solution even when paired with the top $k-1$ edges and the best possible scenario for adding that edge (note we use MPC), then it cannot be part of the optimal solution. Thus, all such edges can safely be removed to reduce the search space.

\begin{figure*}[t]
    \centering
    \includegraphics[width=0.99\textwidth]{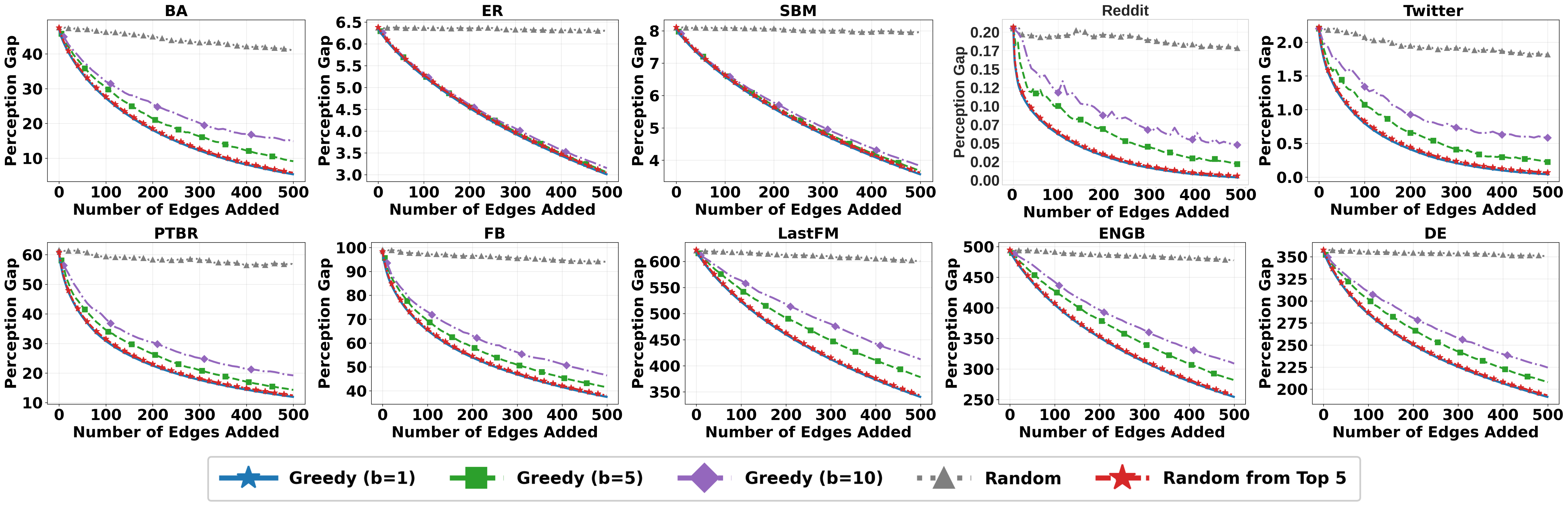} 
    \caption{\textcolor{black}{Value of the perception gap by adding different numbers of edges for the studied algorithms (Random, Greedy, Batch Greedy, and Random Batch Greedy) on our real-world and synthetic networks.}
    } \label{fig:results}
\end{figure*}

\begin{algorithm}[ht]\small
\caption{MPC ($G$, $k$, $\mathbf{S}$)}
\label{mpc_alg}
\begin{algorithmic}[1]
\Require Graph $G=(V, E)$, opinion vector $\mathbf{S}$, integer $k \le |\mathcal{E}_C|$
\Ensure Reduced Search Space
\For{all $e= \{u, v\} \in \mathcal{E}_C$}
    \State Initialize $mpc_e \gets -\infty$
    \For{$k' = 0$ to $k-1$}
        \For{each combination $C$ of $k'$ edges from $\mathcal{E}_C$ incident to $v$ or $u$, except $e=\{u,v\}$} 
            \State Compute perception gap $\P_{\text{orig}}$ for $G'=(V,E \cup C)$
            \State Compute perception gap $\P_{\text{new}}$ for $G''=(V,E \cup C \cup \{ e \})$
            \State $f_e \gets \P_{\text{orig}} - \P_{\text{new}}$
            \If{$f_e > mpc_e$}
                \State $mpc_e \gets f_e$
            \EndIf
        \EndFor
    \EndFor
\EndFor

\State Sort the edges by $mpc_e$ in non-increasing order as $[mpc_1, mpc_2, \dots, mpc_{|\mathcal{E}_C|}]$
\State Set $T = \text{Greedy}(G, k, b=1, \mathbf{S})$
\For{each $e \in \mathcal{E}_C$}
    \State $S_e = mpc_e + \sum_{i=1}^{k-1} mpc_{i}$
    \If{$S_e < \P (G,\mathbf{S}) - \P (G+T,\mathbf{S})$}
        \State Remove edge $e$ from the list of missing edges
    \EndIf
\EndFor
\State \Return the remaining missing edges as the search space
\end{algorithmic}
\end{algorithm}

\textcolor{black}{
\textbf{Complexity of Algorithm~\ref{mpc_alg}.} For a candidate edge $e = \{u,v\}$, let $\mathcal{I}_e$ be the set of missing edges incident to $u$ or $v$ (excluding $e$). In the worst case, $|\mathcal{I}_e| \le 2n - 4 = O(n)$. For each $k' \in \{0,\dots,k-1\}$, the algorithm enumerates all $\binom{|\mathcal{I}_e|}{k'}$ combinations $C$ of $k'$ edges from $\mathcal{I}_e$, which is $O(n^{k'})$. Summing over $k'$, the per-edge enumeration evaluates $O(n^{k-1})$ combinations (dominated by $k' = k-1$ for constant $k$). Since adding $e$ only modifies the degrees and local opinion sums of $u$ and $v$, the marginal change $f_e = \mathcal{P}_{\text{orig}} - \mathcal{P}_{\text{new}}$ for each combination $C$ can be evaluated in $O(1)$ time via Lemma~\ref{Evolution_thorem} after updating the endpoint quantities for $G \cup C$ in $O(k) = O(1)$ time. Thus, evaluating a single candidate edge takes $O(n^{k-1})$ time. Repeating this for all $|\mathcal{E}_C|$ candidate edges yields a total enumeration time of $O(|\mathcal{E}_C| \cdot n^{k-1})$. Since $|\mathcal{E}_C| \le \binom{n}{2} = O(n^2)$, the overall complexity simplifies to $O(n^{k+1})$.
}

Algorithm~\ref{mpc_alg} (MPC) reduces the search space. Then, we can apply a brute-force approach to find the optimal solution. As we will see, this is quite effective and permits us to cover larger graphs in our experiments.

\section{
Experiments 
on Algorithms} \label{experiments_section_4}
\subsection{Experimental Setup} We consider a collection of real-world and synthetic graph data. We use seven real-world datasets: Reddit ($546$ nodes, $8,962$ edges), Twitter ($531$ nodes, $3,621$ edges), Twitch-PTBR ($1,912$ nodes, $31,299$ edges), Facebook ($4,039$ nodes, $88,234$ edges), Twitch-ENGB ($7,126$ nodes, $35,324$ edges), Twitch-DE ($9,498$ nodes, $153,138$ edges) and LastFM ($7,624$ nodes, $27,806$ edges). These datasets are publicly available from
{SNAP dataset repository}~\cite{leskovec2016snap}. The synthetic models (all of size $1000$) include Erd\H{o}s--R\'enyi (ER), Barab\'asi--Albert (BA), and stochastic block model (SBM) networks. For the SBM model, the parameters are set to $p=0.05$ and $q=0.03$. For the BA graph, we set $m=4$, and for the ER graph, we use $p=0.05$ for all of our experiments. (The parameters are set to mimic the degree distribution in the real-world networks of similar size). 
Initial opinions for Reddit were derived using a linguistic analytics tool, and for Twitter, a sentiment analysis tool tailored for tweets was used to analyze first-hour posts. The datasets were collected by~\cite{chitra2020analyzing}. The histogram plot of these opinions is shown in Figure~\ref{fig:opinion_hists}. For other datasets, our experiments draw opinions uniformly at random from the interval $[-1, 1]$, unless otherwise stated.
\begin{figure}[t]
    \centering
    \includegraphics[width=0.8\columnwidth]{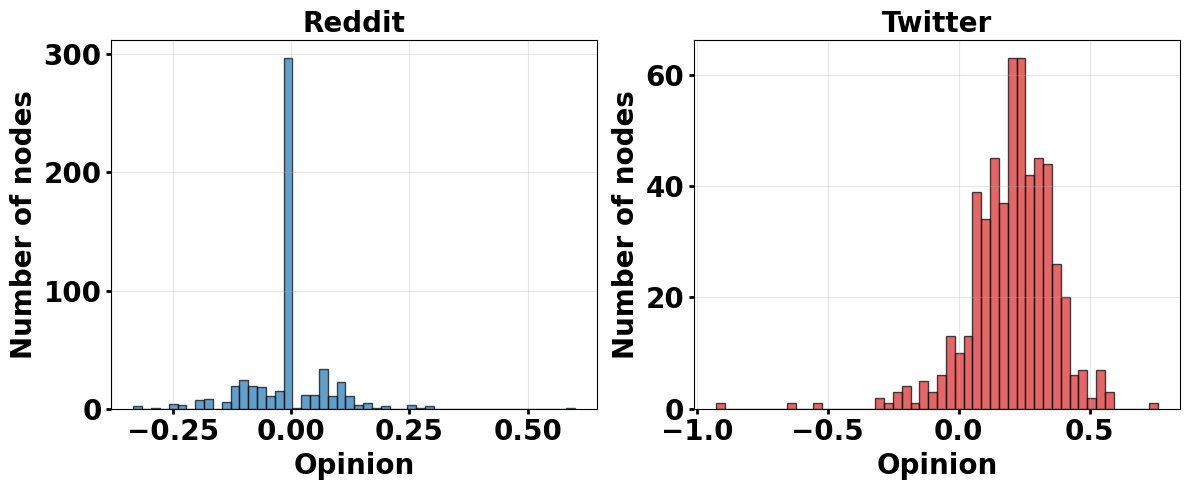}
    \caption{\textbf{Opinion distributions for Reddit and Twitter.} 
    Reddit opinions from linguistic analytics; Twitter from sentiment analysis 
    of first-hour posts~\cite{chitra2020analyzing}.}
    \label{fig:opinion_hists}
\end{figure}

\textcolor{black}{For reproducibility, all stochastic experiments (opinion assignment and, where applicable, random graph generation) are repeated over $10$ independent runs. The base random seed is fixed to $42$ for both NumPy and Python's \texttt{random} module. For each run $r$, we derive a run-specific seed as $42 + r$, applied via \texttt{numpy.random.seed(42+r)} and \texttt{random.seed(42+r)} before that run's random draws. This ensures every run is distinct while the entire experiment remains exactly reproducible from the base seed.
All experiments were implemented in Python~\texttt{3.14.4}, using NumPy~\texttt{2.5.1}, NetworkX~\texttt{3.6.1}, and joblib~\texttt{1.5.3}. All reported wall-clock timings were measured on a single machine with a 13th Gen Intel(R) Core(TM) i7-1365U processor and 32\,GB RAM, running Windows 11 Enterprise (64-bit). Each reported timing is the mean wall-clock time across the independent runs stated for the corresponding experiment.}

\subsection{Comparison of Algorithms}
The performance of the algorithms proposed in Section~\nameref{algorithms} was evaluated on all datasets. The results are illustrated in Figure~\ref{fig:results}. As expected, \textbf{Random} performs poorly in comparison to the greedy-based approaches. Among the greedy-based algorithms, \textbf{Greedy} ($b=1$) performs better than the batch greedy algorithms ($b=5$ and $b=10$) in return for more computation. Overall, these diagrams show that all three algorithms are effective at reducing the perception gap within networks.

\subsection{Evaluating MPC-Based Pruning}
We evaluate the MPC framework along two complementary axes: the end-to-end runtime speedup it delivers over full brute-force search, and how closely the fast greedy algorithm approximates the true optimal solution obtained via MPC-based search.
\textcolor{black}{
\subsubsection{MPC vs Brute-Force (Runtime)}
To evaluate the computational benefit of MPC-based pruning, we conduct an end-to-end runtime comparison between two pipelines that solve the same edge-selection problem exactly: (i) full brute-force search over all missing edges, and (ii) MPC preprocessing, which scores every missing edge's marginal potential contribution and prunes those that cannot improve on a greedy bound, followed by brute-force search restricted only to the surviving edges. Both pipelines use identical exhaustive search over their respective candidate sets, so any difference in runtime is attributable solely to the size of the search space each pipeline explores. We evaluate on ER, BA, and SBM networks at $n \in \{20, 25, 30\}$ for $k=3$ added edges, reporting times averaged over $10$ independent runs (fresh opinion draws, and for BA/SBM fresh graph realizations, per run). Table~\ref{tab:bf_vs_mpc} reports, for each network and size, the fraction of missing edges pruned, the runtime of each pipeline, and the resulting speedup. Across all settings, MPC preprocessing prunes $73$--$96\%$ of missing edges before brute-force search begins, yielding an end-to-end speedup of $9.9\times$--$36.7\times$ over full brute force, with the advantage growing as network size increases. It should be mentioned that MPC pruning is an exact reduction with zero accuracy loss: since MPC uses an upper bound to eliminate only provably suboptimal edges, both pipelines yield the same optimal edge set and identical perception gap reduction. In the following section, we will show that, across diverse networks, there is a high alignment between MPC pruning combined with brute-force search and the greedy Algorithm~\ref{alg:spgreedy}.
\begin{table}[t]
    \centering
\caption{\textbf{End-to-end runtime comparison.} Brute-force (BF) over all missing edges vs. MPC+BF (pruning followed by BF on survivors) for $k=3$, $n=20,25,30$. MPC+BF time is the full pipeline (scoring, pruning, and BF).}
\label{tab:bf_vs_mpc}
    \begin{tabular}{l S[table-format=2.1] S[table-format=3.2] S[table-format=2.2] S[table-format=2.1]}
        \toprule
        \textbf{Network} & \textbf{Pruned (\%)} & \textbf{Full BF (s)} & \textbf{MPC+BF (s)} & \textbf{Speedup} \\
        \midrule
        \rowcolor{gray!15}
        \multicolumn{5}{l}{\textbf{$n=20$}} \\
        ER  & 84.9 & 29.90  & 1.89  & \textbf{15.8$\times$} \\
        BA  & 73.0 & 13.23  & 1.34  & \textbf{9.9$\times$}  \\
        SBM & 81.9 & 35.09  & 2.86  & \textbf{12.3$\times$} \\
        \midrule
        \rowcolor{gray!15}
        \multicolumn{5}{l}{\textbf{$n=25$}} \\
        ER  & 89.7 & 164.53 & 6.67  & \textbf{24.7$\times$} \\
        BA  & 77.0 & 119.57 & 7.45  & \textbf{16.0$\times$} \\
        SBM & 91.6 & 168.30 & 7.06  & \textbf{23.9$\times$} \\
        \midrule
        \rowcolor{gray!15}
        \multicolumn{5}{l}{\textbf{$n=30$}} \\
        ER  & 90.9 & 610.77 & 17.09 & \textbf{35.7$\times$} \\
        BA  & 73.4 & 531.17 & 26.28 & \textbf{20.2$\times$} \\
        SBM & 95.7 & 606.65 & 16.55 & \textbf{36.7$\times$} \\
        \bottomrule
    \end{tabular}
\end{table}
\subsubsection{Greedy vs MPC (Accuracy)}
To evaluate the accuracy of our greedy algorithm, we compare it against the optimal solution obtained via MPC search on networks with $n=100$ nodes. For each run, synthetic networks are independently regenerated with a new seed, while for real-world datasets a fresh subgraph is drawn via random breadth-first search (BFS) sampling. This ensures each run reflects independent graph structure and opinions, rather than repeated sampling on a fixed graph. We report the mean absolute error $\mu$ and standard deviation $\sigma$ over these
$20$ independent runs for $1$, $2$, and $3$ edge additions. Table~\ref{table:greedyVersusOptimal} presents the rounded results, along with $95\%$ confidence intervals and $p$-values from paired $t$-tests for the 3-edge case, which
we focus on as the most meaningful setting where cumulative errors emerge. 
For each network, let $g_r$ and $o_r$ denote the perception gap achieved by Greedy and by the optimal (MPC for pruning and then brute-force on the reduced search space) solution, respectively, in run $r$, for
$r=1,\dots,20$. We test whether Greedy and Optimal differ systematically using a
paired $t$-test on the differences $d_r = g_r - o_r$. The test statistic is
\begin{equation*}
t = \frac{\bar{d}}{s_d/\sqrt{20}}, \ \
\bar{d} = \frac{1}{20}\sum_{r=1}^{20} d_r, \ \
s_d = \sqrt{\frac{1}{20-1}\sum_{r=1}^{20}(d_r - \bar{d})^2},
\end{equation*}
which follows a Student's $t$-distribution under the
null hypothesis $H_0: \mathbb{E}[d_r] = 0$. The resulting $p$-value gives the
probability of observing a mean difference at least as extreme as $\bar{d}$ if
Greedy and Optimal were in fact equivalent in expectation; $p > 0.05$ indicates
we fail to reject $H_0$, i.e., we find no statistically detectable difference
between Greedy and Optimal at the $5\%$ level.
The greedy algorithm achieves near-zero error across all networks
($\mu \leq 0.034$ for the 3-edge case), indicating that in practice its solutions
are close to optimal. Paired $t$-tests yield $p > 0.05$ for all networks; we
therefore fail to reject $H_0$ in every case, finding no statistically detectable
difference between Greedy and Optimal. Together with the small observed effect
sizes, these results are consistent with the greedy algorithm closely
approximating the optimal solution across both synthetic and real-world network
structures, supporting its use for larger graphs where brute-force search is
infeasible.
}
\begin{table}[t]
\centering
\caption{Comparison of \textbf{Greedy} vs.\ optimal solutions with $1$, $2$, and $3$ edge additions. Values are mean $\pm$ std over 20 independent runs; $95\%$ CI in brackets for the 3-edge case. $p$-values from paired $t$-tests; $p > 0.05$ indicates no statistically significant difference.}
\label{table:greedyVersusOptimal}
\setlength{\tabcolsep}{2.5pt}
\renewcommand{\arraystretch}{1.15}
\begin{tabular}{lcccccc}
\toprule
\textbf{Network} & \textbf{1-edge} & \textbf{2-edges} & \textbf{3-edges} & \textbf{$p$} \\
\midrule
\rowcolor{gray!8}  ER      & 0.000$\pm$0.000 & 0.013$\pm$0.028 & 0.034$_{[-0.009,0.077]}$ & 0.092 \\
                   BA      & 0.000$\pm$0.000 & 0.000$\pm$0.001 & 0.003$_{[-0.004,0.009]}$ & 0.371 \\
\rowcolor{gray!8}  SBM     & 0.000$\pm$0.000 & 0.010$\pm$0.013 & 0.031$_{[-0.031,0.094]}$  & 0.237 \\
                   Reddit  & 0.000$\pm$0.000 & 0.001$\pm$0.002 & 0.002$_{[-0.001,0.005]}$ & 0.182 \\
\rowcolor{gray!8}  Twitter & 0.000$\pm$0.000 & 0.001$\pm$0.003 & 0.020$_{[-0.020,0.060]}$ & 0.245 \\
                   Twitch-PTBR    & 0.000$\pm$0.000 & 0.000$\pm$0.000 & 0.002$_{[-0.004,0.008]}$ & 0.374 \\
\rowcolor{gray!8}  LastFM  & 0.000$\pm$0.000 & 0.008$\pm$0.012 & 0.008$_{[-0.006,0.023]}$ & 0.194 \\
                   Facebook      & 0.000$\pm$0.000 & 0.000$\pm$0.000 & 0.000$_{[-0.000,0.001]}$ & 0.374 \\
\rowcolor{gray!8}  Twitch-ENGB    & 0.000$\pm$0.000 & 0.005$\pm$0.011 & 0.005$_{[-0.008,0.019]}$ & 0.329 \\
                   Twitch-DE      & 0.000$\pm$0.000 & 0.000$\pm$0.000 & 0.001$_{[-0.001,0.003]}$ & 0.374 \\
\bottomrule
\end{tabular}
\end{table}

\begin{figure*}[t]
\begin{center}
\includegraphics[scale=0.6]{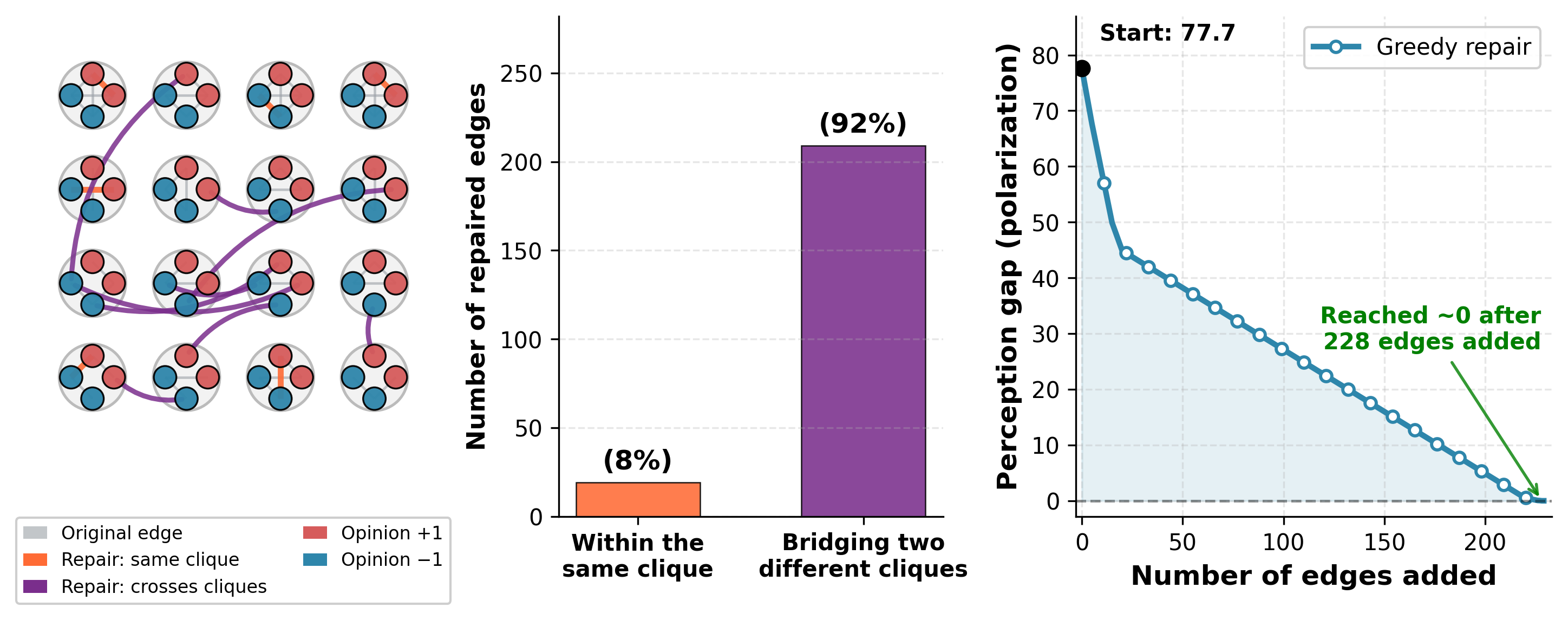} 
\caption{Greedy algorithm on $250$ cliques (four nodes each, two $+1$, two $-1$). After removing $300$ edges, only $228$ added edges eliminate the perception gap. (Right) PGI vs.\ added edges. (Left) Visualization with added edges highlighted.}\label{fig:polarization_near_zer}
\end{center}
\end{figure*}

\begin{figure}[t]
    \centering
    \includegraphics[width=0.5\textwidth]{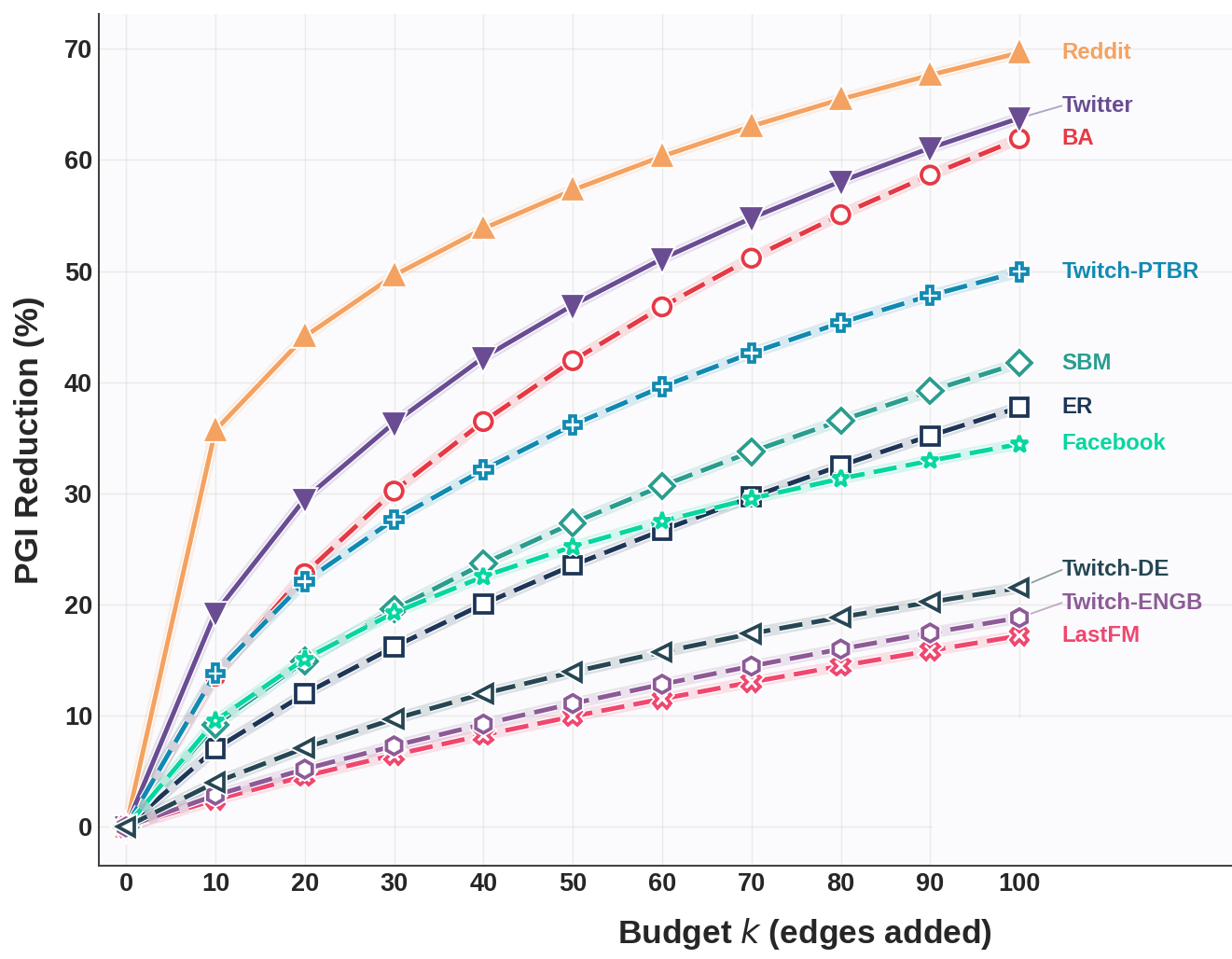}
\caption{PGI reduction vs.\ budget $k$: reduction increases monotonically with added edges across all networks.}
    \label{fig:sensitivity_k}
\end{figure}

As another test (for the purpose of comparing the greedy and the optimal solution), we consider a structure with $M=250$ cliques, each of size four, where two nodes have an opinion of $+1$, and the other two have an opinion of $-1$. The initial perception gap is zero. After removing $k=300$ edges, the \textbf{Greedy} algorithm is applied to add back edges and minimize the perception gap. 
As seen in Figure~\ref{fig:polarization_near_zer} (right), the \textbf{Greedy} algorithm successfully reaches the optimal perception gap of $0$ after adding only $228$ edges, which is $76\%$ of the initially removed edges. This is because the \textbf{Greedy} algorithm takes advantage of connecting different cliques, not just restoring the removed edges. Notably, only $8\%$ of the added edges are within the same clique, while the remaining $92\%$ bridge different cliques, effectively reducing the perception gap by creating cross-clique connections. The graph obtained after adding edges is visualized in Figure~\ref{fig:polarization_near_zer} (left).

\textcolor{black}{
\subsection{Sensitivity Analysis}
We conduct a sensitivity analysis to examine the influence of various parameters on perception gap reduction. In this part, we set $k=100$, and all other experimental settings are as described previously, unless otherwise stated.
}
\textcolor{black}{
\subsubsection{Sensitivity to Budget $k$.}
To evaluate how the reduction in perception gap scales with the number of edges added, we vary the budget $k \in \{0, 10, 20, \dots, 100\}$ on all considered networks. We apply the greedy edge addition algorithm (Algorithm~\ref{alg:spgreedy}, $b=1$) and report the relative PGI reduction $(\mathcal{P}_{\text{orig}} - \mathcal{P}_{\text{final}}) / \mathcal{P}_{\text{orig}} \times 100\%$. 
As shown in Figure~\ref{fig:sensitivity_k}, the PGI reduction increases smoothly with $k$ across all networks. Smaller networks (synthetic, Twitter, Reddit) achieve the highest reductions, while larger networks (LastFM, Twitch-DE and Twitch-ENGB) show more modest reductions. Nonetheless, the reduction increases consistently with $k$ across all datasets, confirming the broad effectiveness of greedy edge addition.
}

\textcolor{black}{
\subsubsection{Sensitivity to Opinion Distribution}
To assess robustness to the underlying opinion assignment, we test four distributions: 
(i) uniform distribution on $[-1,1]$, representing completely random opinions with no inherent bias; 
(ii) Gaussian $\mathcal{N}(0,0.5)$ clipped to $[-1,1]$, capturing moderate opinions centered around zero with limited variance; 
(iii) polarized opinions ($\pm 1$), where half the nodes receive $+1$ and the remaining half receive $-1$, representing extreme ideological division; and 
(iv) Beta distribution Beta$(0.5, 0.5)$ scaled to $[-1,1]$, where we first generate values in $[0,1]$ and then transform them to $[-1,1]$ via $y = 2x - 1$. This generates a U-shaped distribution with most opinions concentrated near the extremes but with continuous variation within each group.
We evaluate these distributions on BA and SBM graphs, averaging over $10$ trials. Figure~\ref{fig:opinion_sensitivity} reports the mean PGI reduction for each distribution.
On both graphs, uniform, Gaussian, and bimodal Beta achieve comparable reductions ($41.2\%$--$43.4\%$ on BA, $17.9\%$--$20.2\%$ on SBM), while the polarized distribution yields the lowest reduction ($24.9\%$ and $8.6\%$, respectively). These results confirm that the greedy algorithm is robust across diverse opinion regimes, with polarized opinions representing a worst-case that still achieves substantial mitigation.
\begin{figure}[t]
    \centering
    \includegraphics[width=0.50\textwidth]{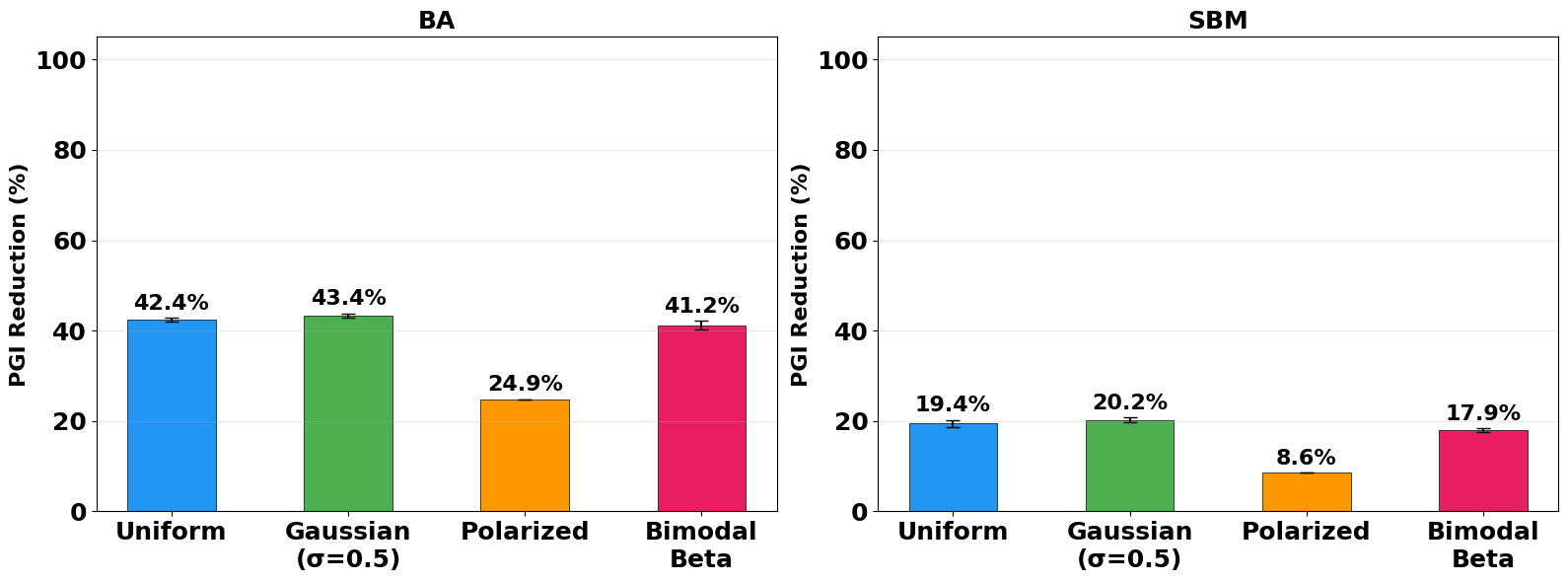}
\caption{Sensitivity to opinion distribution: PGI reduction (\%) at $k=100$ on BA and SBM. The greedy algorithm is consistent across distributions; polarized opinions are the hardest case.}
\label{fig:opinion_sensitivity}
\end{figure}
}
\textcolor{black}{
\subsubsection{Sensitivity to Network Regime}
\label{sec:network_sensitivity}
To evaluate robustness across network structures, we test nine graph variants ($n=1000$) spanning three model families: SBM with varying homophily ($\Delta p = p - q \in \{0.02, 0.05, 0.10\}$), where we keep $S=p+q$ approximately constant (at $S=0.16$) while varying $p-q$, so that the overall expected density remains comparable across regimes\footnote{Expected node degree in a two-block SBM depends on $p+q$, not on $p$ and $q$ individually; fixing this sum isolates homophily strength as the only varying factor. Concretely, $p = (S+\Delta p)/2$ and $q = (S-\Delta p)/2$, giving $(p, q) \in \{(0.09, 0.07),\ (0.105, 0.055),\ (0.13, 0.03)\}$ for $\Delta p \in \{0.02, 0.05, 0.10\}$ respectively.}; ER with varying density ($p \in \{0.03, 0.06, 0.10\}$); and BA with varying preferential attachment ($m \in \{2, 5, 10\}$), averaged over $10$ trials ($k=100$ greedy edges added per trial). For the SBM regimes, opinions are set deterministically by block membership ($+1$/$-1$) to directly test the echo-chamber mechanism (Theorem~\ref{polarization_SBM}); for ER and BA, opinions are drawn independently and uniformly from $[-1,1]$ per trial. Figure~\ref{fig:network_sensitivity} reports the mean PGI reduction for each regime.
The greedy algorithm achieves consistent reduction across different regimes. Performance is highest on BA graphs ($49.1\%$--$31.9\%$), followed by ER ($21.0\%$--$12.7\%$), and lowest on SBM ($7.4\%$--$1.8\%$), which uses the polarized opinion distribution. \textit{Within each family, reduction decreases with density (ER, BA) and decreases with homophily (SBM).}
These results confirm that the greedy algorithm is effective across diverse network regimes.
\begin{figure}[t]
    \centering
    \includegraphics[width=0.40\textwidth]{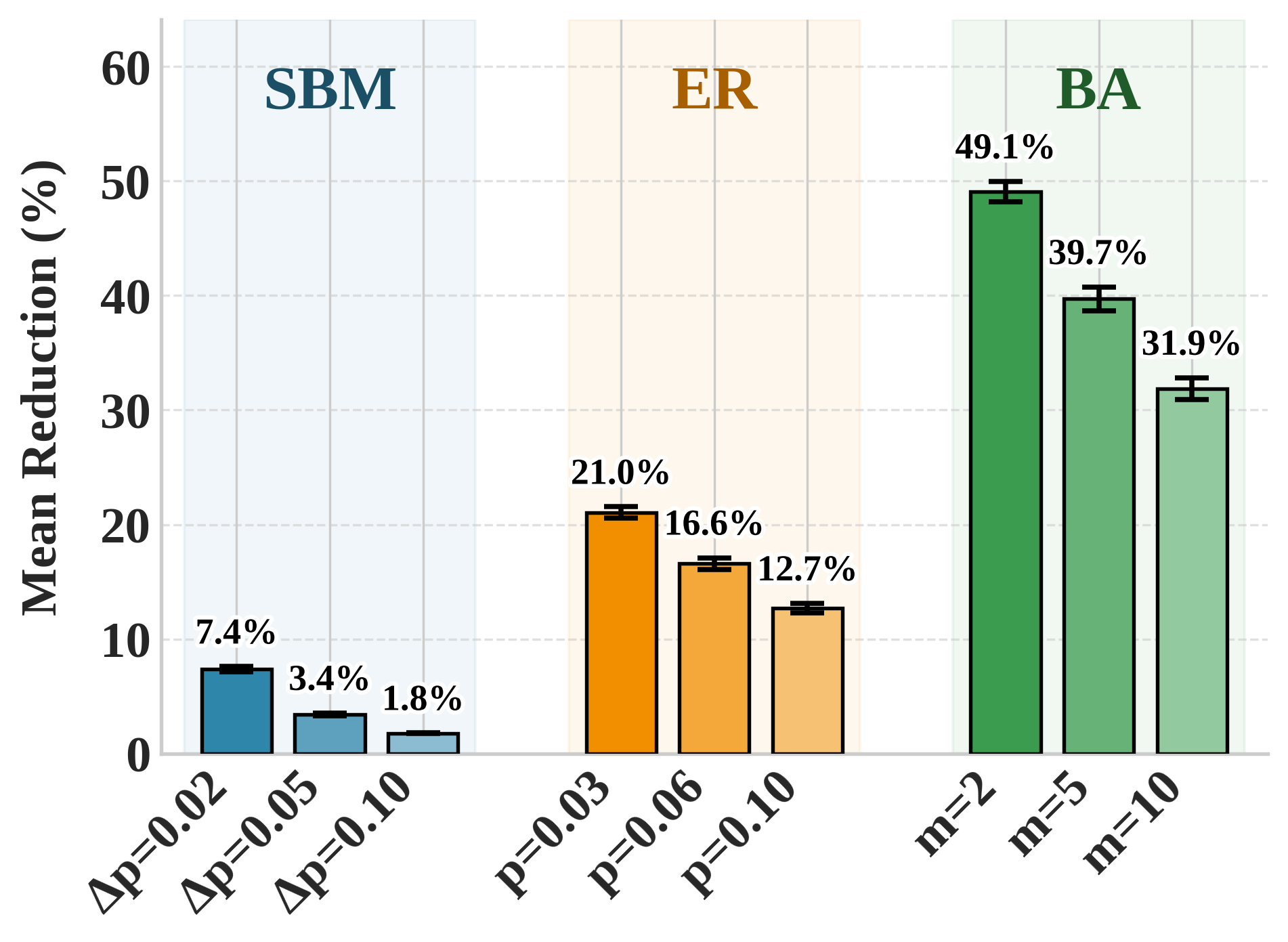}
\caption{Sensitivity to network regime. PGI reduction is consistent across all variants, with BA achieving the highest reduction, followed by ER and SBM.}\label{fig:network_sensitivity}
\end{figure}
}

\newcommand{\dec}[1]{\textcolor{red!70!black}{$\downarrow$#1\%}}
\newcommand{\inc}[1]{\textcolor{black!55!green}{$\uparrow$#1\%}}
\newcommand{\nochg}{0.0\%}

\begin{table}[t]
    \centering
    \caption{\textbf{Side effects of PGI minimization.} Percentage change in PGI, modularity, majority illusion fraction, and network disagreement after adding $k=100$ greedy edges. $\downarrow$ (red) = reduction; $\uparrow$ (green) = increase.}
    \label{tab:side_effects}
    \setlength{\tabcolsep}{6pt}
    \begin{tabular}{lcccccc}
        \hline
\small\textbf{Network} & \small\textbf{PGI} & \small\textbf{Mod.} & \small\textbf{Maj. Illus} & \small\textbf{Disagr.} \\
\hline
        BA     & \dec{42.1} & \dec{8.3}  & \inc{0.4} & \inc{5.5} \\
        ER  & \dec{16.5} & \dec{5.7}  & \nochg    & \inc{1.0} \\
        SBM          & \dec{17.7} & \dec{2.6}  & \nochg    & \inc{0.9} \\
        Reddit       & \dec{69.6} & \dec{6.4}  & \dec{75.0} & \inc{13.1} \\
        Twitter      & \dec{63.8} & \dec{10.4} & \nochg    & \inc{27.1} \\
        Twitch-PTBR  & \dec{50.6} & \dec{4.1}  & \dec{1.4} & \inc{0.7} \\
        LastFM       & \dec{16.1} & \dec{1.6}  & \nochg    & \inc{1.3} \\
        Facebook     & \dec{33.6} & \dec{1.8}  & \dec{0.3} & \inc{0.3} \\
        Twitch-ENGB  & \dec{18.1} & \dec{1.6}  & \nochg    & \inc{0.9} \\
        Twitch-DE    & \dec{20.4} & \dec{1.1}  & \nochg    & \inc{0.2} \\
        \hline
        \textbf{Mean} & \textbf{\dec{34.8}} & \textbf{\dec{4.4}} & \textbf{\dec{7.6}} & \textbf{\inc{5.1}} \\
        \hline
    \end{tabular}
\end{table}

\begin{table*}[t]
\centering
\caption{Batch-size comparison: percentage reduction in the perception gap relative to the initial (pre-repair) value, and wall-clock runtime. For $b=5$ and $b=10$, parenthesised values show the change in reduction (percentage points) and the speedup factor relative to $b=1$.}
\label{tab:batch_reduction}
\renewcommand{\arraystretch}{1.2}
\begin{tabular}{l
    S[table-format=2.1] S[table-format=3.1]
    c c
    c c}
\toprule
& \multicolumn{2}{c}{\textbf{$b=1$}} & \multicolumn{2}{c}{\textbf{$b=5$}} & \multicolumn{2}{c}{\textbf{$b=10$}} \\
\cmidrule(lr){2-3} \cmidrule(lr){4-5} \cmidrule(lr){6-7}
\textbf{Dataset} & {\textbf{Red.\ (\%)}} & {\textbf{Time (s)}} & \textbf{Red.\ (\%)} & \textbf{Time (s)} & \textbf{Red.\ (\%)} & \textbf{Time (s)} \\
\midrule
\rowcolor{gray!8}  Reddit      & 69.6 & 1.0   & 51.2 (\textcolor{red!70!black}{$-$18.5}) & 0.3 (\textcolor{green!55!black}{3.0$\times$}) & 42.4 (\textcolor{red!70!black}{$-$27.3}) & 0.2 (\textcolor{green!55!black}{5.2$\times$}) \\
                   Twitter     & 63.8 & 1.0   & 51.3 (\textcolor{red!70!black}{$-$12.5}) & 0.3 (\textcolor{green!55!black}{3.1$\times$}) & 39.1 (\textcolor{red!70!black}{$-$24.7}) & 0.2 (\textcolor{green!55!black}{5.9$\times$}) \\
\rowcolor{gray!8}  Twitch-PTBR   & 50.9 & 12.1  & 43.7 (\textcolor{red!70!black}{$-$7.2})  & 4.1 (\textcolor{green!55!black}{3.0$\times$}) & 37.9 (\textcolor{red!70!black}{$-$12.9}) & 2.2 (\textcolor{green!55!black}{5.5$\times$}) \\
                   Facebook    & 34.4 & 53.2  & 28.8 (\textcolor{red!70!black}{$-$5.6})  & 20.9 (\textcolor{green!55!black}{2.5$\times$}) & 23.5 (\textcolor{red!70!black}{$-$10.9}) & 11.2 (\textcolor{green!55!black}{4.7$\times$}) \\
\rowcolor{gray!8}  Twitch-ENGB & 19.1 & 187.7 & 12.9 (\textcolor{red!70!black}{$-$6.2})  & 82.8 (\textcolor{green!55!black}{2.3$\times$}) & 10.6 (\textcolor{red!70!black}{$-$8.5})  & 63.1 (\textcolor{green!55!black}{3.0$\times$}) \\
                   Twitch-DE   & 20.8 & 349.1 & 15.3 (\textcolor{red!70!black}{$-$5.5})  & 153.6 (\textcolor{green!55!black}{2.3$\times$}) & 13.0 (\textcolor{red!70!black}{$-$7.9})  & 81.9 (\textcolor{green!55!black}{4.3$\times$}) \\
\rowcolor{gray!8}  LastFM & 17.2 & 206.3 & 11.4 (\textcolor{red!70!black}{$-$5.8})  & 95.9 (\textcolor{green!55!black}{2.2$\times$}) & 9.6 (\textcolor{red!70!black}{$-$7.6})   & 49.8 (\textcolor{green!55!black}{4.1$\times$}) \\
                   ER          & 18.6 & 3.6   & 18.3 (\textcolor{red!70!black}{$-$0.3})  & 1.2 (\textcolor{green!55!black}{2.9$\times$}) & 17.3 (\textcolor{red!70!black}{$-$1.4})  & 0.7 (\textcolor{green!55!black}{4.9$\times$}) \\
\rowcolor{gray!8}  BA          & 42.5 & 3.7   & 36.6 (\textcolor{red!70!black}{$-$5.8})  & 1.3 (\textcolor{green!55!black}{2.9$\times$}) & 32.7 (\textcolor{red!70!black}{$-$9.7})  & 0.7 (\textcolor{green!55!black}{5.2$\times$}) \\
                   SBM         & 18.5 & 3.6   & 17.8 (\textcolor{red!70!black}{$-$0.7})  & 1.2 (\textcolor{green!55!black}{2.9$\times$}) & 17.5 (\textcolor{red!70!black}{$-$1.0})  & 0.7 (\textcolor{green!55!black}{5.1$\times$}) \\
\bottomrule
\end{tabular}
\end{table*}

\textcolor{black}{
\subsubsection{Comparison with Realistic Link Recommendation Baselines}
To contextualize the performance of the greedy algorithm, we compare it against three practical link recommendation baselines: \textbf{Common-Neighbors} (add edges between nodes with the most shared neighbors), \textbf{Jaccard} (highest Jaccard coefficient), and \textbf{Random} (uniformly random edge addition). We evaluate on two representative networks, BA as a synthetic benchmark and Facebook as a real-world case, with uniform opinions and budgets $k \in \{0, 10, \dots, 100\}$.
Figure~\ref{fig:baseline_comparison} reports the PGI reduction for all methods. On both networks, the greedy algorithm achieves substantial and steadily increasing reduction (reaching over $30\%$ for Facebook and $40\%$ for BA at $k=100$), while Common-Neighbors, Jaccard, and Random yield less than $5\%$ in all cases. This is not surprising, as they connect structurally similar nodes without regard for opinion discrepancies. 
This stark contrast demonstrates that generic link prediction heuristics are ineffective for mitigating the perception gap, and that optimizing directly for PGI reduction is essential for meaningful intervention.
}
\begin{figure}[t]
    \centering
    \includegraphics[width=0.7\textwidth]{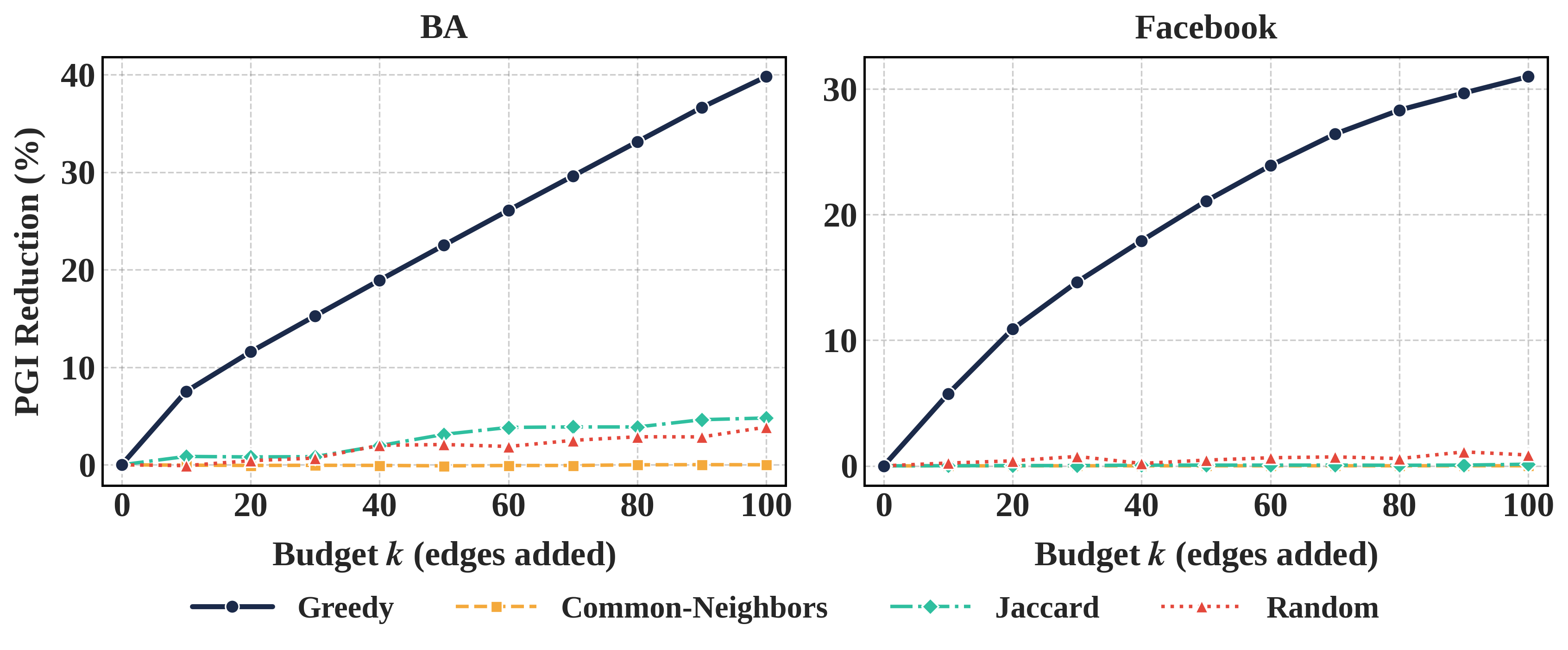}
    \caption{Comparison with realistic baselines. PGI reduction (\%) on BA and Facebook. Greedy substantially outperforms Common-Neighbors, Jaccard, and Random, which achieve negligible reduction on both networks.}
    \label{fig:baseline_comparison}
\end{figure}

\textcolor{black}{
\subsubsection{Side Effects of PGI Minimization}
Beyond PGI reduction, we examine the side effects of greedy edge addition on modularity with respect to the opinion-sign partition, majority illusion fraction, and network disagreement across all networks. Modularity is computed using the sign of opinions as community labels~\cite{newman2006modularity}, capturing the strength of community separation.
Majority illusion measures the fraction of nodes whose neighbors' majority label (by sign) differs from the global majority label~\cite{lerman2016majority}. Specifically, each opinion is first binarized by sign, $b_i = \mathrm{sign}(s_i)$; the global majority is $\mathrm{sign}(\sum_i b_i)$, the local majority of node $i$ is $\mathrm{sign}(\sum_{j \in N[i]} b_j)$, and a node is under illusion if these labels differ. 
Network disagreement is defined as $\sum_{(u,v) \in E} (s_u - s_v)^2$, where the sum is over all edges in the final graph. To consider an extreme case, in this experiment, opinions are drawn from a Beta$(0.5, 0.5)$ distribution rescaled to $[-1,1]$, which is U-shaped and concentrates mass near the extremes. This simulates an already-polarized opinion landscape rather than moderate or uniformly spread opinions.
Table~\ref{tab:side_effects} reports the percentage change in each metric. PGI reduction ($16.1\%$--$69.6\%$, mean $34.8\%$) is accompanied by decreased modularity with respect to the opinion-sign partition (mean $4.4\%$) and decreased majority illusion with mean $7.6\%$, except BA, which is slightly positive. This improvement, however, comes with a trade-off: network disagreement increases on every network (mean increase of $5.1\%$), and particularly strongly on some networks, because the greedy objective preferentially connects nodes with differing opinions to pull local neighborhood averages toward the global mean. Modularity decreases because greedy edges connect nodes across opinion groups, weakening community separation. Majority illusion changes minimally (except for Reddit) because it is a binary metric that only flips when a node's local majority label crosses the tie threshold, whereas PGI is continuous and captures even small shifts in local averages toward the global mean. The dramatic reduction in majority illusion for Reddit ($-75.0\%$) is explained by its global majority of $0$ (see Figure~\ref{fig:opinion_hists}): greedy edges shift local majorities from $\pm1$ to $0$, directly matching the global majority. We also note that variance-based polarization remains unchanged here, as edge additions do not alter the fixed opinion distribution. These results indicate that PGI minimization yields complementary reductions in modularity and majority illusion, but at the cost of increased network disagreement.
}
\textcolor{black}{
\subsubsection{Batch Size Sensitivity}
To evaluate the trade-off between computational efficiency and solution quality, we compare three batch sizes: $b=1$ (standard greedy), $b=5$, and $b=10$. For each setting, we report the percentage reduction in the perception gap relative to the initial (pre-repair) value, the wall-clock runtime, and, for $b=5$ and $b=10$, the change in reduction (in percentage points) and the speedup factor relative to $b=1$.
Table~\ref{tab:batch_reduction} summarizes the results. Unsurprisingly, $b=1$ achieves the highest reduction across all datasets, as it selects the single best edge at each step. Larger batch sizes offer substantial speedups: on real-world networks, $b=5$ achieves $2.2$--$3.1\times$ speedup, while $b=10$ achieves $3.0$--$5.9\times$ speedup. The reduction loss varies by dataset; for instance, on Facebook, $b=10$ yields a $4.7\times$ speedup with a $10.9$ percentage point loss ($34.4\%$ vs.\ $23.5\%$), while on ER and SBM the loss is minimal ($<2$ percentage points). Overall, $b=5$ offers a favorable trade-off between runtime and solution quality for most applications, whereas $b=10$ is suitable when runtime is the primary constraint and some reduction can be sacrificed.
}
\section{Conclusion}\label{conclusion}
This paper introduces the PGI, a measure that complements existing metrics such as variance-based polarization, which often overlook structural dependencies, and extends phenomena like the majority illusion to continuous opinion settings, thereby offering a more nuanced understanding of perception disparities in relation to network structure.
We derived theoretical bounds showing how structural properties, such as eigenvalues in regular graphs and inter-group connectivity in stochastic block models, affect the perception gap. We also studied the problem of minimizing the perception gap through targeted link recommendations, proving its inapproximability and lack of monotonicity and supermodularity in the objective function. Despite these challenges, we proposed several greedy-based algorithms that were shown to be very effective on real-world and synthetic data.

Future work includes exploring alternative formulations of the perception gap that admit approximation guarantees, extending the framework to weighted networks and opinion dynamics models such as the FJ model, and developing stochastic analyses with variance and concentration bounds.

\bibliographystyle{unsrt}
\bibliography{references}
\end{document}